\documentclass[12pt]{amsart}
\usepackage[hmargin=0.8in,height=8.8in]{geometry}
\usepackage{graphicx}
\usepackage{amssymb,amsthm, times, mathptmx}
\usepackage{delarray,verbatim}
\usepackage[font=scriptsize]{caption}
\usepackage[font=scriptsize]{subcaption}
\usepackage{srcltx}
\usepackage{bigints}
\usepackage[font=scriptsize]{caption}

\usepackage{xcolor}
\usepackage{soul}

\usepackage[numbers]{natbib}
\usepackage{ifpdf}
\usepackage{color}
\definecolor{webgreen}{rgb}{0,.5,0}
\definecolor{webbrown}{rgb}{.8,0,0}
\definecolor{emphcolor}{rgb}{0.95,0.95,0.95}

\usepackage{hyperref}
\hypersetup{%
          colorlinks=true,
          linkcolor=black,
          filecolor=black,
          citecolor=black,
          breaklinks=true}
\ifpdf \hypersetup{pdftex,
            pdfstartview=FitH, 
            bookmarksopen=true,
            bookmarksnumbered=true
} \else \hypersetup{dvips} \fi

\newcommand {\bracks}[1]{\left[#1\right]}
\newcommand {\parens}[1]{\left(#1\right)}

\numberwithin{equation}{section}

\newtheorem{proposition}{Proposition}[section]

\newtheorem{remark}{Remark}[section]
\newtheorem{lemma}{Lemma}[section]

\newtheorem{assump}{Assumption}[section]

\numberwithin{remark}{section} \numberwithin{proposition}{section}
\numberwithin{corollary}{section}
\renewcommand{\S}{\mathcal{S}}

\newcommand {\R}{\mathbb{R}}

\newcommand {\F}{\mathcal{F}}

\newcommand {\p}{\mathbb{P}}
\newcommand {\G}{\mathfrak{G}}
\newcommand {\E}{\mathbb{E}}

\newcommand{\N}{\mathbb{N}}

\newcommand{\diff}{{\rm d}}

\newcommand{\word}{\hspace{0.2cm}}
\newcommand{\conn}{\quad\text{and}\quad}
\newcommand{\1}{\mbox{1}\hspace{-0.25em}\mbox{l}}
\newcommand{\lev}{L\'{e}vy }

\newcommand{\II}{\mathcal{I}}
\newcommand{\bw}{\Bigm|}

\newcommand{\nn}{\nonumber}

\begin{document}
\title{A Direct Solution Method for Pricing  Options in regime-switching models}
\author[Egami]{Masahiko Egami}
\author[Kevkhishvili]{RUSUDAN KEVKHISHVILI}
\thanks{First draft: November 24, 2017; this version: \today. \\ An early version of this work is circulated under the title ``On the optimal stopping problem of linear diffusions in regime-switching models".  The second author is in part supported by JSPS KAKENHI Grant Number JP 17J06948.}
\date{}
\maketitle
\begin{abstract}
Pricing financial or real options with arbitrary payoffs in regime-switching models is an important problem in finance.  Mathematically, it is to solve, under certain standard assumptions, a general form of optimal stopping problems in regime-switching models. In this article, we reduce an optimal stopping problem with an arbitrary value function in a two-regime environment  to a pair of optimal stopping problems without regime switching. We then propose a method for finding optimal stopping rules using the techniques available for non-switching problems. In contrast to other methods, our systematic solution procedure is more direct since we first obtain the explicit form of the value functions. In the end, we discuss an option pricing problem which may not be dealt with by the conventional methods, demonstrating the simplicity of our approach.
\end{abstract}
\noindent \small{Key Words: Optimal stopping, Markov switching, Diffusion, Concavity, Perpetual Options	
\newline \noindent JEL Classification: C610, C630, G130, G300
}

\noindent\textbf{MASAHIKO EGAMI:} Graduate School of Economics,
Kyoto University, Sakyo-Ku, Kyoto, 606-8501, Japan.\\
Email: \email{egami@econ.kyoto-u.ac.jp}

\noindent\textbf{RUSUDAN KEVKHISHVILI:} Graduate School of Economics, Kyoto University, Sakyo-Ku, Kyoto, 606-8501, Japan. Research Fellow of Japan Society for the Promotion of Science.\\
Email: \email{keheisshuiri.rusudan.73m@st.kyoto-u.ac.jp}

\section{Introduction}\label{sec:introduction}
Regime-switching models have been of interest to economists as early as the 1970's, since the parameters describing economic data exhibit changes over periods of time (see, e.g., \citet{palgrave_dictionary}). For this reason, the regime switching has been incorporated into optimal stopping problems as well, allowing the underlying diffusion to have different parameters that depend on the regime. Such optimal stopping problems have various applications, such as option pricing and real options analysis in corporate finance.
\newline \indent One way of analyzing the value function of the optimal stopping problem with regime-switching is its characterization as a  viscosity solution to the associated  Hamilton–-Jacobi-–Bellman equation. For example, \citet{pemy2014} uses this method for the finite time horizon optimal stopping problem of regime-switching \lev process, while \citet{liu2016} analyses the infinite time horizon state-dependent regime-switching problem. Such studies use numerical schemes to demonstrate the application of the results, but do not offer a universal method to find the value function explicitly. The numerical methods for the evaluation of the value function are discussed in \citet{boyarchenko2009}, \citet{labahn2011}, \citet{babbin2014}, as well.
Another approach is to first guess a solution by constructing an optimal strategy and then prove that the candidate function is indeed a value function (e.g., see \citet{guo2004}). This approach is sometimes called the ``guess and verify" method.  The success of this method depends on the underlying diffusion process and the reward function of the problem; therefore, there are problems to which this method cannot be applied.  Hence the remaining issue is to derive a universal expression of the value function and suggest a technique for explicitly evaluating it for general classes of reward functions and underlying diffusions.

In this paper, we study an infinite time horizon regime-switching optimal stopping problem of an arbitrary continuous reward function that satisfies the linear growth condition, a standard assumption in the literature (see, e.g., \citet[Sec. 5.2]{Pham-book}). We reduce this problem to a pair of non-switching ones and evaluate the value functions explicitly.
It is worth discussing the idea here.  Suppose that we have two regimes $i=1$ and $i=2$. The state variable $X$ switches from $X^{(1)}$ to $X^{(2)}$ and vice versa (see Section \ref{sec:math-prep} for rigorous problem formulation). We first reduce the problem to a couple of one-dimensional problems without switching (subsection \ref{subsec:reduction}). It follows that, for each regime, the state space $\II$ is split into continuation region $\mathrm{C}_i$ and stopping region $\Gamma_i$, $i=1, 2$.  Then, defining $A_1:=\Gamma_1 \cap \Gamma_2,\; A_2:=\mathrm{C}_1 \cap \Gamma_2,\; A_3:=\mathrm{C}_2 \cap \Gamma_1, \; A_4:=\mathrm{C}_1\cap \mathrm{C}_2$, we can write
\begin{align*}
\II=(\Gamma_1 \cap \Gamma_2) \cup (\mathrm{C}_1 \cap \Gamma_2)\cup (\mathrm{C}_2 \cap \Gamma_1)\cup (\mathrm{C}_1\cap \mathrm{C}_2)=A_1 \cup A_2 \cup A_3\cup A_4.
\end{align*}
Next, (1) we characterize the value functions of the problem in an explicit way (see subsections \ref{subsec:form} and \ref{subsec:ode-solution}). With the form of the value functions available, we propose a procedure for solving the optimal stopping problem in Section \ref{sec:method}. Specifically, based on the result (1), (2) we obtain the value functions in each region $A_1$ through $A_4$ as the smallest nonnegative concave majorants in some transformed space: the method for non-switching problems presented in \citet{DK2003}. Thanks to (1) and (2), the proof of optimality is greatly reduced: we simply need to check certain geometric properties of the associated functions.  Therefore, our method is rather direct and differs from so-called ``guess and verify" methods conventional in the literature of regime-switching optimal stopping. Due to the direct nature of our solution, it is free from hard analysis that often requires ingenuity for proving some inequalities which are referred to as the \emph{verification lemma} (see, e.g., \citet[Theorem 10.4.1]{oksendal-book}). Moreover, the proposed method works for any diffusion whose parameters satisfy Lipschitz and linear growth conditions (Section \ref{subsec:regime-switching}) and any continuous reward function satisfying the abovementioned condition. Hence, it can apply to problems that may not be handled by the conventional methods.  We demonstrate this in Section  \ref{sec:example}  with a worked example of capped-call options in the regime-switching model.

It should be noted that the reduction of the regime-switching problem to non-switching ones is also discussed in \citet{le-wang}. They consider a finite time horizon optimal stopping problem for a continuous non-increasing nonnegative convex payoff function with regime-switching geometric Brownian motion (GBM) and break the problem into a sequence of optimal stopping problems of constant-coefficient GBM processes. \citet{le-wang} derive an iterative method for approximating the value function of the original problem. Their method differs from ours in that we transform the original regime-switching optimal stopping problem into a pair of non-switching ones and find a solution explicitly. We emphasize that this is a novel approach. In our study we also treated a broader class of reward functions and underlying diffusions.
\newpage
\section{Mathematical Framework}\label{sec:math-prep}
\subsection{Basic facts}\label{subsec:basic-fact}
Let us consider the probability space $(\Omega, \F, \p)$, where $\Omega$ is the set of all possible realizations of the
stochastic economy, and $\p$ is a probability measure defined on $\F$. We denote by
$\mathbb{F}=\{\F_t\}_{t\ge 0}$ the filtration satisfying the usual conditions and consider a regular time-homogeneous diffusion process $X$ adapted to $\mathbb{F}$. Let $\p_x$ denote the probability measure associated to $X$ when started at $x\in \mathcal{I}$.
The state space of $X$ is $\mathcal{I}=(\ell, r)\in \R$, where $\ell$ and $r$ are both \emph{natural} boundaries.  That is, $X$ cannot start from or exit from $\ell$ or $r$.
We assume that $X$ satisfies the following stochastic differential equation:
\[\diff X_t=\mu (X_t)\diff t + \sigma (X_t) \diff B_t\conn X_0=x,\]
where $B=\{B_t: t\ge 0\}$ is a standard Brownian motion and  $\mu:\mathcal{I}\rightarrow\R$ and $\sigma : \mathcal{I}\rightarrow\R_+$ satisfy the following Lipschitz and linear growth conditions (see \citet[Theorem 2.9, Chapter 5]{karatzas-shreve-book1})
\begin{align}\label{eq:Lipschitz-and-growth}
|\mu(x)-\mu(y)|+|\sigma(x)-\sigma(y)|&\leq \bar{K}|x-y| \quad \quad \forall x,y\in\mathcal{I}\\
\mu(x)^2+\sigma(x)^2&\leq K\parens{1+x^2}\nonumber
\end{align}
with positive constants $\bar{K}$ and $K$. The second condition is implied by the first one. These conditions ensure the existence and uniqueness of a strong solution given an initial condition.
The optimal stopping problem in one dimension is formulated by
\begin{equation}\label{eq:no-regime}
V(x)=\sup_{\tau\in \S}\E^x[e^{-q \tau}h(X_\tau)]
\end{equation}
where $ q>0$,  $\tau\ge0$, $\S$ is a set of stopping times of $\mathbb{F}$ and $h:(\ell,r)\to\R$ is a continuous Borel function on $\mathcal{I}$.
\newline \indent To review the general theory for optimal stopping of linear diffusions,
we state some fundamental facts: the infinitesimal generator $\G$ for a continuous function $v$ of the process $X$ is defined by
\[\G v(\cdot)=\lim\limits_{t\downarrow 0}\frac{1}{t}\left(P_tv(\cdot)-v(\cdot)\right)\]
where $P_t$ is the transition semigroup of $X$. The equation $\G v-qv=0$ has two
fundamental solutions: $\psi_{q}(\cdot)$ and $\varphi_{q}(\cdot)$.
We set
$\psi_q(\cdot)$ to be the increasing and $\varphi_q(\cdot)$ to be the
decreasing solution.  They are linearly independent positive
solutions and are uniquely determined up to multiplication. It is well
known that for $H_z:=\inf\{t\ge 0: X_t=z\}$,
\begin{align*}
\E^x\bracks{e^{-q H_z}}
=\begin{cases}
\frac{\psi_q(x)}{\psi_q(z)}, & x \le z,\\[4pt]
\frac{\varphi_q(x)}{\varphi_q(z)}, &x > z.
\end{cases}
\end{align*} For the complete characterization of $\psi(\cdot)$ and
$\varphi(\cdot)$, refer to \citet[Section 4.6]{IM1974}. Let us now define
\begin{align} \label{eq:F}
F(x)&:=\frac{\psi_q(x)}{\varphi_q(x)}, \hspace{0.5cm} x\in
\mathcal{I}.
\end{align}
Then $F(\cdot)$ is continuous and strictly increasing.  Next,
following \citet{dynkin} (Section 8 in Appendix), we define concavity
of a function with respect to $F$ as follows:
A real-valued function $u$ is called \emph{$F$-concave} on $\mathcal{I}$
if for any $[c,d]\subseteq \mathcal{I}$, we have for every $x\in[c, d]$ 
\[
u(x)\geq
u(c)\frac{F(d)-F(x)}{F(d)-F(c)}+u(d)\frac{F(x)-F(c)}{F(d)-F(c)}.\]
When both boundaries $\ell$ and $r$ are natural, $V(x)<+\infty$ for all $x\in (\ell, r)$ if and only if
\begin{align}\label{eq:finiteness}
\xi_{\ell}:=\limsup_{x\downarrow \ell}\frac{h^+(x)}{\varphi_q(x)} \conn \xi_r:=\limsup_{x\uparrow r}\frac{h^+(x)}{\psi_q(x)}
\end{align}
are both finite (\citet[Proposition 5.10]{DK2003}). Here $h^+(x)=\max\{h(x),0\}$.
Let $W:[0,\infty)\to\R$ be the \emph{smallest nonnegative concave majorant} of
\begin{equation}\label{eq:transform}
H(y):=\begin{cases}
\frac{h}{\varphi_q}\circ F^{-1}(y) \quad \word \text{if} \word y>0,\\
\xi_{\ell} \qquad \qquad \quad\;\:\text{if} \word y=0,
\end{cases}
\end{equation}
where $F^{-1}$ is the inverse of $F$. Then we have
\begin{equation}\label{eq:v=phiW}
V(x)=\varphi_q(x)W(F(x)) \quad \text{for all} \word x\in(\ell,r)
\end{equation}
and finally,  when $\xi_{\ell}=\xi_r=0$, the optimal stopping and continuation regions are
\[
\Gamma:=\{x\in (\ell, r): V(x)=h(x)\} \conn \mathrm{C}:=\{x\in (\ell, r): V(x)>h(x)\}.
\] with the optimal stopping time $\tau^*:=\inf\{t\ge 0: X_t\in \Gamma\}$. See \citet[Propositions 5.12 and 5.13]{DK2003}.
Note that for the rest of this article, the term ``transformation" should be understood as \eqref{eq:transform}.

\subsection{Regime-switching models}\label{subsec:regime-switching}
Given the probability space, the diffusion in this article is of the form
\begin{equation*}
\diff X_t=\mu(X_t, \eta_t)\diff t + \sigma (X_t, \eta_t) \diff B_t,\quad X_0=x,
\end{equation*}
where $\eta_t\in \{1, 2\}$ is a two-state continuous-time Markov chain  adapted to $\mathbb{F}$, independent of $B$, and has a generator of the form
\[\left(
\begin{array}{cc}
-\lambda_1 & \lambda_1 \\
\lambda_2 & -\lambda_2 \\
\end{array}
\right),\] 
with $\lambda_i>0,\; i=1, 2$.
We assume that $\mu(x, i)$ and $\sigma(x, i)$ satisfy the Lipschitz and linear growth conditions \eqref{eq:Lipschitz-and-growth} for each $i$. 
Then, there exists a unique strong solution $X_t$ to the regime-switching differential equation above (\citet[Theorem 3.13]{markovian_sde}).
For brevity we will call the diffusion with drift $\mu(\cdot, i)$ and diffusion parameter $\sigma(\cdot, i)$, $i=1, 2,$ \emph{diffusion $X^{(i)}$} when no confusions arise.  Then our problem is to solve
\begin{equation}\label{eq:problem}
v^*(x, i):=\sup_{\tau\in \S}\E[e^{-q \tau}h(X_\tau)|X_0=x, \eta_0=i]=\sup_{\tau\in \S}\E_{i}^x[e^{-q \tau}h(X_\tau)], \quad i=1, 2.
\end{equation}
Let the stopping and continuation regions be $\Gamma_i$ and $\mathrm{C}_i$, respectively.
To obtain a concrete result useful in solving real-life problems, we set the following assumptions:
\begin{assump}\label{assump1}\normalfont\hspace{1cm}
\begin{enumerate}
\item The reward $h$ is a  continuous function that satisfies the linear growth condition: $|h(x)|\le C(1+|x|)$ on $ \mathcal{I}$ for some strictly positive constant $C<\infty$.
\item The discount rate $q>\beta$, where $\beta$ is a constant satisfying $x\mu(x,i)+\frac{1}{2}\sigma(x,i)^2\le\beta(1+x^2) \quad  \forall (x,i)\in\mathcal{I}\times\{1,2\}$ (see \citet[Lemma 5.2.1]{Pham-book} and \citet[Lemma 3.1]{pham2007} for a similar statement).
\end{enumerate}
\end{assump}
\noindent Note that our assumptions on the parameters $\mu(\cdot,i)$ and $\sigma(\cdot,i)$ guarantee the existence of such $\beta$. In addition, for each $i=1,2$, we have
\begin{equation}\label{eq:bounds}
\E^x_i[X_t^2]\le (1+x^2)e^{2\beta t} \quad \forall t\ge0
\end{equation}
from It\^{o}'s formula and Gronwall's lemma.
The rest of this section is a review of some useful results that relate to non-switching problem \eqref{eq:no-regime} because we can apply them to our regime-switching problem \eqref{eq:problem}.
When Assumption \ref{assump1} is satisfied for the non-switching problem, we have the following proposition:
\begin{proposition}\label{prop:pham}
	Under Assumption \ref{assump1}, the value function $V$ in \eqref{eq:no-regime} is a unique solution to
	\begin{equation}\label{eq:vi}
	\min[qV-\G V, \word V-h]=0
	\end{equation}
	satisfying the linear growth condition.
\end{proposition}
\begin{proof}
	Equation \eqref{eq:bounds} holds in this case and we have $\E^x[e^{-qt}|X_t|]\le e^{-qt}(1+|x|)e^{\beta t}\le 1+|x| \quad \forall t\ge0$. This implies the linear growh of $V$, since
	$V\le\E^x[\sup\limits_{t\ge 0}e^{-qt}h(X_t)]\le C\E^x[\sup\limits_{t\ge 0}e^{-qt}(1+|X_t|)]\le C+C(1+|x|)$ with some constant $C>0$. The proof for \eqref{eq:vi} is immediate from Theorem 5.2.1, Lemma 5.2.2 and Remark 5.2.1 in \citet{Pham-book}.
\end{proof}
For a real value $q>0$, let us recall the resolvent operator
\begin{equation}\label{eq:resolvent}
U^{(q)}f(x):=
\E^x\left[\int_0^\infty e^{-qt}f(X_t)\diff t\right]
\end{equation}
for a continuous function $f$. Thanks to Proposition \ref{prop:pham} together with Assumption \ref{assump1}, we ensure  that
\begin{equation}\label{eq:Uh-finite}
U^{(q)}h(x)<\infty \quad \text{and} \quad U^{(q)}V(x)<\infty \quad x\in \II
\end{equation} for $V$ in \eqref{eq:no-regime}.
In the next section,  we will solve the optimal stopping problem in regime-switching model \eqref{eq:problem} under Assumption \ref{assump1} with the aid of the $F$-concavity characterization and the variational inequality \eqref{eq:vi} for the value function.

\section{Characterization of the value function}\label{sec:formulation}
In this section, we will obtain an explicit form of the value functions of \eqref{eq:problem} in steps. In subsection \ref{subsec:reduction}, we will rewrite the original problem into a pair of one-dimensional optimal stopping problems, making a quantitative analysis of the value function.  We then split the state space $\mathcal{I}$ into regions, depending on which regime the state variable lives in and on whether a point in $\mathcal{I}$ belongs to continuation region or stopping region. We provide, for each region, the form of the value functions.  This is done in subsection \ref{subsec:form}. Finally, in subsection \ref{subsec:ode-solution} we analyze the ODE that arises in one of the regions.

\subsection{Reduction to a pair of one-dimensional problems}\label{subsec:reduction}
We mainly consider $v^*(x, 1)=\sup_{\tau\in\S}\E^x_1[e^{-q\tau}h(X_\tau)]$ since $v^*(x, 2)$ can be handled in the same way.  Let us define a Poisson process $N^1=(N^{1}_t)_{t\ge 0}$ with rate $\lambda_1$ independent of $B$. The first arrival time of $N^1$, denoted by $T^1$, has the exponential distribution whose density is
\[
p(t)=\begin{cases}
\lambda_1 e^{-\lambda_1 t}, & t\ge0,\\
0, & \text{otherwise}.
\end{cases}
\]
Similarly, we define $T^2$ as the first arrival time of a Poisson process $N^2$ with rate $\lambda_2$ independent of $B$ and $N^1$. Then, $T^2$ is the exponential random variable with parameter $\lambda_2$ independent of $B$ and $T^1$.
In the sequel, we transform our original problem \eqref{eq:problem} into a pair of optimal stopping problems without regime-switching. For this, we introduce the notation $\bar{v}(\cdot,j)$ which indicates the optimal value obtained when starting at $X_{T^i}$, $i\neq j$ (see Lemma \ref{lem:hitting}). Later in Proposition \ref{prop:v-star-version} we will prove that $\bar{v}(\cdot,j)$ is the same as $v^*(\cdot,j)$.
\begin{lemma}\label{lem:hitting}
For each $i$, the value function in the regime-switching problem \eqref{eq:problem} is continuous, satisfies the linear growth condition and the following dynamic programming principle (DPP):
\begin{equation}\label{eq:dpp}
v^*(x,i)=\sup\limits_{\tau\in\S}\E_i^x[e^{-q\tau}h(X_{\tau})\1_{\tau<T^i}+e^{-qT^i}\bar{v}(X_{T^i},j)\1_{\tau>T^i}] \qquad i\neq j, \; i,j=1,2,
\end{equation}
where $\bar{v}(x,j)=\sup\limits_{\tau\in \S}\E_j^x[e^{-q\tau}h(X_{\tau})]$ with $x$ evaluated at $X_{T^i}$. \\
Moreover, for each $i$, the supremum in \eqref{eq:problem} over $\tau\in\mathcal{S}$ is attained when $\tau$ is a hitting time of the set $\{x\in\mathcal{I}: v^*(x,i)=h(x)\}$. 
\end{lemma}
\begin{proof}
The proof of the linear growth condition is identical to the one in Proposition \ref{prop:pham}. In fact, from \eqref{eq:bounds}, we have 
\[\E_1^x[e^{-qt}|X_t|]\le e^{-qt}\sqrt{\E_1^x[(X_t)^2]}  \le e^{-qt}e^{\beta t}(1+|x|)\le(1+|x|) \quad \forall t\ge 0.\] 
Therefore, due to the linear growth of $h$,
\[v^*(x,1)\le \E_1^x[\sup\limits_{t\ge 0}e^{-qt}h(X_t)]\le C\E_1^x[\sup\limits_{t\ge 0}e^{-qt}(1+|X_t|)]\le C+C\E_1^x[\sup\limits_{t\ge 0}e^{-qt}|X_t|]\le C+C(1+|x|)\]
with constant $C>0$ which proves the claim. \\
Next, we prove the continuity of $v^*(x,1)$. Due to the Lipschitz and linear growth conditions on the diffusion parameters for each regime, the following holds for all $t\ge0$ (see \citet[eq. (1.19)]{Pham-book}):
\begin{align}\label{eq:pham-1.19}
\E_1[\sup\limits_{0\le u\le t}|X^x_u-X^y_u|]\le \sqrt{\E_1[\sup\limits_{0\le u\le t}|X^x_u-X^y_u|^2]}\le e^{\beta_0t}|x-y|
\end{align}
where $\beta_0$ satisfies $(\mu(x,i)-\mu(y,i))(x-y)+\frac{1}{2}(\sigma(x,i)-\sigma(y,i))^2\le \beta_0|x-y|^2$,  $\forall i$ and $\forall x,y\in\mathcal{I}$. The proof of \eqref{eq:pham-1.19} is an immediate application of It\^{o}'s formula and Gronwall's lemma and its technique is the same for both regime-switching and non-switching diffusion: the superscripts $x$ and $y$ denote the starting positions. Given that the initial regime is 1, take any sequence $(x_n)_{n\in\N}$ in $\mathcal{I}$ converging to $x$ and consider random variables $Y_n:=\sup\limits_{0\le u\le t}|h(X_t^{x_n})-h(X_t^x)|$ for any fixed $t$. From \eqref{eq:pham-1.19}, we know that $\sup\limits_{0\le u\le t}|X^{x_n}_u-X^x_u|\to 0$ in probability for all $t\ge 0$ as $n\to\infty$. Since $h$ is continuous, we have $|h(X_u^{x_n})-h(X_u^x)|\overset{p}{\to}0$ for any $0\le u\le t$ using \citet[Chapter 3, Proposition 3.5]{cinlar}. 
Therefore, $Y_n\to 0$ in probability as $n\to \infty$. Furthermore, since $h$ satisfies the linear growth condition, 
\begin{align*}
\E_1[Y_n^2]&=\E_1[\sup\limits_{0\le u\le t}|h(X_u^{x_n})-h(X_u^x)|^2]
\le 3C\E_1[\sup\limits_{0\le u\le t}(1+|X_u^{x_n}|)^2]+3C\E_1[\sup\limits_{0\le u\le t}(1+|X_u^{x}|)^2]
\end{align*}
with constant $C>0$. We have $\sup\limits_n\E_1[Y_n^2]<\infty$ due to
$\E_1[\sup\limits_{0\le u\le t}|X_u^x|^2]\le (1+2x^2)e^{C't}$ for any $x\in\mathcal{I}$ with constant $C'$ depending on Lipschitz and linear growth conditions of parameters (\citet[Thm 3.24]{markovian_sde}).
This implies that $(Y_n)$ is a uniformly integrable sequence and therefore, converges to $0$ in $L^1$ since it also converges to $0$ in probability. Since we used an arbitrary $t$ in the definition of $Y_n$, we obtain
$\lim\limits_{n\to\infty}\E_1[\sup\limits_{t\ge 0}e^{-qt}|h(X_{t}^{x_n})-h(X_{t}^x)|]=0.$
Then, 
\begin{align*}
\lim\limits_{n\to\infty}|v(x_n,1)-v(x,1)|&\le \lim\limits_{n\to\infty}\sup\limits_{\tau\in\S}\E_1[e^{-q\tau}|h(X_{\tau}^{x_n})-h(X_{\tau}^x)|]\le\lim\limits_{n\to\infty}\E_1[\sup\limits_{t\ge 0}e^{-qt}|h(X_{t}^{x_n})-h(X_{t}^x)|]=0
\end{align*}
which proves that $v(x,1)$ is continuous on $\mathcal{I}$.\\ 
Finally, we prove the dynamic programming principle. Since we know that the value functions satisfy the linear growth condition, we have $\bar{v}(X_{T^i},j)<\infty$ for each $j\neq i$; therefore, for any $X_{T^i}$, we can take an $\epsilon$-optimal stopping time $\tau^{\varepsilon}\in\S$ such that 
\[\E_j^{X_{T^i}}[e^{-q\tau^{\varepsilon}}h(X_{\tau^{\varepsilon}})]\ge \bar{v}(X_{T^i},j)-\varepsilon\]
with $\varepsilon>0$.
To indicate that the choice of $\tau^{\varepsilon}$ depends on the starting regime and state, we write $\tau^{\varepsilon}(X_{T^i},j)$. 
Using the shift operator $\theta$, set $\tau':=\tau\1_{\tau<T^1}+(T^1+\tau^{\varepsilon}(X_{T^1},2)\circ\theta_{T^1})\1_{\tau>T^1}$, where $\tau$ is an $\F$-stopping time and $\tau^{\varepsilon}$ denotes the chosen $\varepsilon$-optimal stopping time when the position $X_{T^1}$ is known. That is, 
\[\E_2^{X_{T^1}}[e^{-\tau^{\varepsilon}(X_{T^1},2)q}h(X_{\tau^{\varepsilon}(X_{T^1},2)})]\ge \bar{v}(X_{T^1},2)-\varepsilon \quad \text{almost surely}.\]
Since $\{T^1\le t\}\in \mathcal{F}_t$, $\{\tau'\le t\}\in\mathcal{F}_t$ and $\tau'$ is an $\F$-stopping time. Then, we have
\begin{align*}
v^*(x,1)&\ge \E_1^x[e^{-q\tau'}h(X_{\tau'})]=\E_1^x[e^{-q\tau}h(X_{\tau})\1_{\tau<T^1}+e^{-q(T^1+\tau^{\varepsilon}(X_{T^1},2)\circ\theta_{T^1})}h(X_{T^1+\tau^{\varepsilon}(X_{T^1},2)\circ\theta_{T^1}})\1_{\tau>T^1}]\\
&=\E_1^x[e^{-q\tau}h(X_{\tau})\1_{\tau<T^1}+e^{-qT^1}\E^{X_{T^1}}_2[e^{-q\tau^{\varepsilon}(X_{T^1},2)}h(X_{\tau^{\varepsilon}(X_{T^1},2)})]\1_{\tau>T^1}]\\
&\ge\E_1^x[e^{-q\tau}h(X_{\tau})\1_{\tau<T^1}+e^{-qT^1}\bar{v}(X_{T^1},2)\1_{\tau>T^1}]-\varepsilon
\end{align*}
where we used strong Markov property in the second line. Since $\tau$ was an arbitrary stopping time, we obtain
\begin{equation}\label{eq:DPP>}
v^*(x,1)\ge \sup_{\tau\in \S}\E_1^x[e^{-q\tau}h(X_{\tau})\1_{\tau<T^1}+e^{-qT^1}\bar{v}(X_{T^1},2)\1_{\tau>T^1}].
\end{equation}
Next, we have 
\begin{align*}
\E_1^x[e^{-q\tau}h(X_{\tau})\1_{\tau<T^1}+e^{-q\tau}h(X_{\tau})\1_{\tau>T^1}]&=\E_1^x[e^{-q\tau}h(X_{\tau})\1_{\tau<T^1}+\1_{\tau>T^1}\E_1[e^{-q\tau}h(X_{\tau})\mid \F_{T^1}]]\\
&=\E_1^x[e^{-q\tau}h(X_{\tau})\1_{\tau<T^1}+\1_{\tau>T^1}e^{-qT^1}\E_1[e^{-q(\tau-T^1)}h(X_{\tau})\mid \F_{T^1}]]\\
&\le \E_1^x[e^{-q\tau}h(X_{\tau})\1_{\tau<T^1}+\1_{\tau>T^1}e^{-qT^1}\bar{v}(X_{T^1},2)]
\end{align*}
due to the definition of $\bar{v}$. Therefore, we obtain
\[v^*(x,1)\le  \sup_{\tau\in \S}\E_1^x[e^{-q\tau}h(X_{\tau})\1_{\tau<T^1}+e^{-qT^1}\bar{v}(X_{T^1},2)\1_{\tau>T^1}]\]
This, together with \eqref{eq:DPP>}, proves the dynamic programming principle \eqref{eq:dpp}. The linear growth condition, continuity and the DPP hold similarly for $v^*(x,2)$. \\
We have shown that $v^*(x,i)$ is a solution to one-dimensional optimal stopping problem \eqref{eq:dpp} from which it is obvious that for each $i$, the optimal $\tau$ is a hitting time of the set $\{x\in\mathcal{I}: v^*(x,i)=h(x)\}$ by invoking \citet[Proposition 5.13 ]{DK2003}.
\end{proof}
In Lemma \ref{lem:hitting} we have shown that the supremum over $\tau$ can be taken over first hitting times of $X$. We proceed further 
by using the resolvent operator defined in \eqref{eq:resolvent} (note that $\bar{v}(\cdot,j)\ge 0$, $j=1,2$) and set
\begin{align}\label{eq:g12}
\E_1^x\left[e^{-q T^1}\bar{v}\left(X_{T^1}, 2\right)\right]
&=\int_0^{\infty}\E_1^x\left[e^{-q t}\bar{v}(X_{t}, 2)\1_{T^1\in\diff t}\right]=\int_0^{\infty}\E_1^x\left[e^{-q t}\bar{v}(X_{t}, 2)\mid T^1\in\diff t\right]\p_{x}(T^1\in\diff t)\nn\\
&=\int_0^{\infty}\E_1^x[e^{-q t}\bar{v}(X_{t}, 2)]\cdot \p_{x}(T^1\in\diff t)=\E_1^x\left[\int_0^{\infty}e^{-q t}\bar{v}(X_{t}, 2)\lambda_1e^{-\lambda_1t}\diff t\right]\nn\\ &=\E_1^x\left[\lambda_1\int_0^\infty e^{-(q+\lambda_1)t}\bar{v}(X_t, 2)\diff t\right]=\lambda_1U^{(q+\lambda_1)}\bar{v}(x, 2)=:\bar{g}_{1,2}(x)
\end{align}
where we use the independence of $X_t$ and $T^1$ for the third equality.
Similarly, we define
\begin{equation}\label{eq:g21}
\bar{g}_{2,1}(x):=\lambda_2U^{(q+\lambda_2)}\bar{v}(x, 1)=\E_2^x\left[\lambda_2\int_0^\infty e^{-(q+\lambda_2)t}\bar{v}(X_t, 1)\diff t\right]=\E_2^x\left[e^{-q T^2}\bar{v}\left(X_{T^2}, 1\right)\right].
\end{equation}
Since $T^1$ is exponentially distributed with parameter $\lambda_1$, we can further write the first term in \eqref{eq:dpp} as
\begin{align*}
\E_1^x\left[\1_{\{\tau<T^1\}}e^{-q \tau} h(X_{\tau})\right]
&=
\int_{0}^{\infty}\E_1^x\left[\1_{\{\tau<t\}}e^{-q \tau}h(X_{\tau})\1_{T^1\in\diff t}\right]=\int_{0}^{\infty}\E_1^x\left[\1_{\{\tau<t\}}e^{-q \tau}h(X_{\tau})\mid T^1\in \diff t\right]\cdot\p_x(T^1\in\diff t)\\
&=\int_{0}^{\infty}\E_1^x\left[\1_{\{\tau<t\}}e^{-q \tau}h(X_{\tau})\right]\cdot\p_x(T^1\in\diff t)
=\E_1^x\left[\int_0^{\infty}\1_{\{\tau<t\}}e^{-q \tau}h(X_{\tau})\lambda_1e^{-\lambda_1t}\diff t\right]\\
&=\E_1^x\left[\int_{\tau}^{\infty}e^{-q \tau}h(X_{\tau})\lambda_1e^{-\lambda_1t}\diff t\right]=\E_1^x\left[\lambda_1 e^{-q{\tau}}h(X_{\tau})\int_{\tau}^\infty e^{-\lambda_1 t}\diff t\right]=\E_1^x[e^{-(q+\lambda_1){\tau}}h(X^{(1)}_{\tau})]
\end{align*}
with the process $X^{(1)}$ having drift $\mu(t, 1)$ and diffusion parameter $\sigma(t, 1)$  since all the realizations are in the set of $\{\omega\in \Omega: {\tau}(\omega)<T^1(\omega)\}$.
Next, in view of \eqref{eq:g12}, the second term of \eqref{eq:dpp} is simplified to
\begin{align}\label{eq:second_term}
\E_1^x\left[\1_{\{T^1<\tau\}}e^{-qT^1}\bar{v}(X_{T^1}, 2)\right]
&=\int_0^{\infty}\E_1^x\left[\1_{\{t<\tau\}}e^{-qt}\bar{v}(X_{t}, 2)\1_{T^1\in\diff t}\right]=\int_0^{\infty}\E_1^x\left[\1_{\{t<\tau\}}e^{-qt}\bar{v}(X_{t}, 2)\mid T^1\in\diff t\right]\cdot\p_x(T^1\in\diff t)\nonumber\\
&=\int_0^{\infty}\E_1^x\left[\1_{\{t<\tau\}}e^{-qt}\bar{v}(X_{t}, 2)\right]\cdot\p_x(T^1\in\diff t)\nonumber\\
&=\E_1^x\left[\int_0^{\infty}\1_{\{t<\tau\}}e^{-qt}\bar{v}(X_{t}, 2)\lambda_1e^{-\lambda_1t}\diff t\right]=\E_1^x\left[\int_0^{\tau}\lambda_1e^{-(q+\lambda_1)t}\bar{v}(X_{t}, 2)\diff t\right]\nonumber\\
&=\E_1^x\left[\lambda_1 \int_0^\infty e^{-(q+\lambda_1)t}\bar{v}(X_t, 2)\diff t\right]-\E_1^x\left[\lambda_1 \int_{\tau}^\infty e^{-(q+\lambda_1)t}\bar{v}(X_t, 2)\diff t\right]\nn\\
&=\bar{g}_{1,2}(x)-\E_1^x\left[ e^{-(q+\lambda_1){\tau}}\E_1^{X_{\tau}^{(1)}}\left[\lambda_1\int_0^\infty e^{-(q+\lambda_1)t}\bar{v}(X_t, 2)\diff t\right]\right]\nn\\
&=\bar{g}_{1,2}(x)-\E_1^x[e^{-(q+\lambda_1){\tau}}\bar{g}_{1,2}(X^{(1)}_{\tau})].
\end{align}
Note that $\tau<T^1$ in the second term of the fourth row in \eqref{eq:second_term}. In view of \eqref{eq:g12}, $-\bar{g}_{1,2}(X_{\tau}^{(1)})=-\E_1^{X_{\tau}^{(1)}}\left[e^{-qT^1}\bar{v}(X_{T^1},2)\right]$.  
Now, $\bar{g}_{1,2}(X_{\tau}^{(1)})$ is the expectation evaluated  by the position of $X^{(1)}$ with drift $\mu(\cdot, 1)$ and diffusion parameter $\sigma(\cdot, 1)$ at the stopping time $\tau$ and is the forgone value that one would have obtained if (1) one does \emph{not} stop at time $\tau$ and lets it switch to regime 2 at time $T^1$ and (2) behaves optimally from time $T^1$.\\
By combining the first and the second terms of \eqref{eq:dpp} evaluated above and stressing out the fact that this equation is concerned with a hitting time for $X^{(1)}$, we write
\begin{equation*}\label{eq:v-1}
v^*(x, 1)=\sup_{\tau_1\in\S}\E^x_1\bracks{e^{-(q+\lambda_1)\tau_1}(h-\bar{g}_{1, 2})(X_{\tau_1}^{(1)})}+\bar{g}_{1, 2}(x).
\end{equation*}  Similarly, we can derive
\begin{equation*}\label{eq:v-2}
v^*(x, 2):=\sup_{\tau_2\in\S}\E^x_2\bracks{e^{-(q+\lambda_2)\tau_2}(h-\bar{g}_{2, 1})(X_{\tau_2}^{(2)})}+\bar{g}_{2, 1}(x)
\end{equation*} with $\tau_2$ being a hitting time of the process $X^{(2)}$ with drift $\mu(\cdot, 2)$ and diffusion parameter $\sigma(\cdot, 2)$.
Thus far, we have established the coupled optimal stopping problems
\begin{subequations}\label{eq:couple}
	\begin{align}
	v^*(x, 1)&=\sup_{\tau_1\in\S}\E_1^x\left[e^{-(q+\lambda_1)\tau_1}(h-\bar{g}_{1, 2})(X_{\tau_1}^{(1)})\right]+\bar{g}_{1, 2}(x),\label{eq:couple1}\\
	v^*(x, 2)&=\sup_{\tau_2\in\S}\E_2^x\left[e^{-(q+\lambda_2)\tau_2}(h-\bar{g}_{2, 1})(X_{\tau_2}^{(2)})\right]+\bar{g}_{2, 1}(x),\label{eq:couple2}
	\end{align}
\end{subequations}
where $\bar{g}_{1, 2}$ and $\bar{g}_{2, 1}$ are
\begin{equation}\label{eq:g-couple}
\begin{cases}
\bar{g}_{1, 2}(x)=\E_1^x\left[e^{-q T^1}\bar{v}(X_{T^1}, 2)\right]=\lambda_1U^{(q+\lambda_1)}\bar{v}(x, 2),\\
\bar{g}_{2, 1}(x)=\E_2^x\left[e^{-q T^2}\bar{v}(X_{T^2}, 1)\right]=\lambda_2U^{(q+\lambda_2)}\bar{v}(x, 1).
\end{cases}
\end{equation}
See \eqref{eq:g12} and \eqref{eq:g21}.

With this reduction complete, we prove the main result of this subsection.
\begin{proposition}\label{prop:v-star-version}
	The regime-switching problem \eqref{eq:problem} is equivalent to the following pair of one-dimensional problems
	\begin{subequations}\label{eq:sup-couple}
		\begin{align}
		v^*(x, 1)=\sup_{\tau_1\in\S}\E_1^x\left[e^{-(q+\lambda_1)\tau_1}\left(h(X_{\tau_1}^{(1)})-\lambda_1U^{(q+\lambda_1)}v^*(X_{\tau_1}^{(1)}, 2)\right)\right]+\lambda_1U^{(q+\lambda_1)}v^*(x, 2), \label{eq:sup-couple1}\\
		v^*(x, 2)=\sup_{\tau_2\in\S}\E_2^x\left[e^{-(q+\lambda_2)\tau_2}\left(h(X_{\tau_2}^{(2)})-\lambda_2U^{(q+\lambda_2)}v^*(X_{\tau_2}^{(2)}, 1)\right)\right]+\lambda_2U^{(q+\lambda_2)}v^*(x, 1).\label{eq:sup-couple2}
		\end{align}
	\end{subequations}
	In other words, $\bar{v}(x, i), i=1, 2$ in \eqref{eq:g-couple} can be replaced by $v^*(x, i), i=1, 2$.\\ 
	The function value $v^*(X_{\tau_1}^{(1)}, 2)$ in \eqref{eq:sup-couple1} means the following: we stop $X^{(1)}$ at $\tau_1$ and evaluate $v^*(x,2)$ at $x=X_{\tau_1}^{(1)}$. Similarly, the function value $v^*(X_{\tau_2}^{(2)}, 1)$ in \eqref{eq:sup-couple2} means that we stop $X^{(2)}$ at $\tau_2$ and evaluate $v^*(x,1)$ at $x=X_{\tau_2}^{(2)}$.
\end{proposition}
\begin{proof}
	By the symmetry, it suffices to prove the equivalence of  \eqref{eq:couple1} and  \eqref{eq:sup-couple1}.  Set the right-hand side of \eqref{eq:sup-couple1} as $J(x, 1)$.  Let the optimal times of \eqref{eq:couple1} and \eqref{eq:sup-couple1} be $\bar{\tau}_1$ and $\acute{\tau}_1$, respectively. By the definition of $\bar{v}(\cdot, 2)$ we have $J(x, 1)\le v^*(x, 1)$ (see Lemma \ref{lem:hitting}). In this proof, we write $\bar{v}_2(x):=\bar{v}(x, 2)$ and $v^*(x, 2):=v^*_2(x)$ for brevity. We will show $J(x, 1)\ge v^*(x, 1)$:
	\begin{align*}
	&J(x, 1)-v^*(x, 1)\\
	&=\lambda_1U^{(q+\lambda_1)}(v^*_2(x)-\bar{v}_2(x))+\E_1^x\left[e^{-(q+\lambda_1)\acute{\tau}_1}\left(h-\lambda_1U^{(q+\lambda_1)}v_2^*\right)(X_{\acute{\tau}_1}^{(1)})-e^{-(q+\lambda_1)\bar{\tau}_1}\left(h-\lambda_1U^{(q+\lambda_1)}\bar{v}_2\right)(X_{\bar{\tau}_1}^{(1)})\right]\\
	&\ge \lambda_1U^{(q+\lambda_1)}(v^*_2(x)-\bar{v}_2(x))+\E_1^x\left[e^{-(q+\lambda_1)\bar{\tau}_1}\left(h-\lambda_1U^{(q+\lambda_1)}v_2^*\right)(X_{\bar{\tau}_1}^{(1)})-e^{-(q+\lambda_1)\bar{\tau}_1}\left(h-\lambda_1U^{(q+\lambda_1)}\bar{v}_2\right)(X_{\bar{\tau}_1}^{(1)})\right]\\
	&=\lambda_1U^{(q+\lambda_1)}(v^*_2(x)-\bar{v}_2(x))-\lambda_1\E_1^x\left[\int_{\bar{\tau}_1}^\infty e^{-(q+\lambda_1)t} v^*_2(X_t^{(1)})\diff t-\int_{\bar{\tau}_1}^\infty e^{-(q+\lambda_1)t} \bar{v}_2(X_t^{(1)})\diff t\right]
	\end{align*} where the inequality in the third line is due to the definition of $\acute{\tau}_1$ and the equality in the fifth line holds by the strong Markov property. Since $v^*(\cdot, 2)$ is the value function of the original problem \eqref{eq:problem}, we have $v^*(x, 2)\ge \bar{v}(x, 2)$. Therefore, we have
	\begin{align*}
	J(x, 1)-v^*(x, 1)
	&\ge\lambda_1\left(U^{(q+\lambda_1)}(v^*_2(x)-\bar{v}_2(x))-\E_1^x\left[\int_{\bar{\tau}_1}^\infty e^{-(q+\lambda_1)t} (v^*_2(X_t^{(1)})-\bar{v}_2(X_t^{(1)}))\diff t\right]\right)\\
	&=\lambda_1\E_1^x\left[\int_0^{\bar{\tau}_1}e^{-(q+\lambda_1)t}(v^*_2(X_t^{(1)})-\bar{v}_2(X_t^{(1)}))\diff t\right]\ge 0,
	\end{align*}
	which completes the proof.
\end{proof}
\begin{remark}\begin{normalfont}
If we assume that the value of $q$ also changes depending on the regime, Proposition \ref{prop:v-star-version} still holds with $q$ replaced by $q_1$ in \eqref{eq:sup-couple1} and by $q_2$ in \eqref{eq:sup-couple2}, as long as $q_i$ satisfies Assumption \ref{assump1} for each $i$. All the subsequent analysis in the paper also holds with $q$ replaced by $q_1$ and $q_2$ in an obvious way.\end{normalfont}
\end{remark}

\subsection{The form of the value functions}\label{subsec:form}
Now we will investigate  the state space $\II=(\ell, r)$. We define $A_1:=\Gamma_1 \cap \Gamma_2, \; A_2:=\mathrm{C}_1 \cap \Gamma_2, \; A_3:=\mathrm{C}_2 \cap \Gamma_1, \; A_4:=\mathrm{C}_1\cap \mathrm{C}_2$. Then, 
\[\II=(\Gamma_1 \cap \Gamma_2) \cup (\mathrm{C}_1 \cap \Gamma_2)\cup (\mathrm{C}_2 \cap \Gamma_1)\cup (\mathrm{C}_1\cap \mathrm{C}_2)=A_1 \cup A_2 \cup A_3\cup A_4.\]
Below, we write down the form of the value functions in each region. Each $A_j$ may be disconnected; however, the form of the value functions derived below remains true in all parts of $A_j$. Let us denote by $\G_i, i=1, 2$ the generators with drift  $\mu(\cdot, i)$ and diffusion parameter $\sigma(\cdot, i)$. First, note that \eqref{eq:sup-couple} has to be solved simultaneously and that the value functions $v^*(x, i)$ have different forms in $A_j, j=1, 2, 3, 4$. Consequently, the reward functions in \eqref{eq:sup-couple} are not continuous on $\II$ but are piecewise continuous on each $A_j$. Using this fact, we will prove the following proposition:

\begin{proposition}\label{prop:regions}
	In region $A_1=\Gamma_1\cap \Gamma_2$, the value functions of \eqref{eq:sup-couple} are $v^*(x, 1)= v^*(x, 2)=h(x)$.  In region $A_2=\mathrm{C}_1 \cap \Gamma_2$, the value functions are
\[\begin{cases}v^*(x, 1)=u_1(x)+\lambda_1 U^{(q+\lambda_1)}h(x),\\
v^*(x, 2)=h(x),\end{cases}\]
	where $u_1$ satisfies \begin{equation}\label{eq:qvi-1}
	(q+\lambda_1 -\G_1)u_1 =0.
	\end{equation}   For region $A_3=\mathrm{C}_2 \cap \Gamma_1$, we replace the roles of $i=1$ and $i=2$ and $u_2$ satisfies \begin{equation}\label{eq:qvi-2}
	(q+\lambda_2-\G_2)u_2 =0.
	\end{equation}
	In region $A_4=\mathrm{C}_1\cap \mathrm{C}_2$, the value functions are
	\[
	\begin{cases}v^*(x, 1)=\hat{u}_1(x)+\lambda_1 U^{(q+\lambda_1)}v_2^p(x),\\
	v^*(x, 2)=\hat{u}_2(x)+\lambda_2 U^{(q+\lambda_2)}v_1^p(x),
	\end{cases}
	\] where $\hat{u}_i$ solves $(q+\lambda_i-\G_i)\hat{u}_i=0$, $i=1,2$, $v_1^p(x)$ satisfies
	\[(q+\lambda_2-\G_2) (q+\lambda_1-\G_1)v_1^p(x)=\lambda_1 \lambda_2 v_1^p(x),\] and $v_2^p(x)$ satisfies
	\[  (q+\lambda_1-\G_1) (q+\lambda_2-\G_2)v_2^p(x)=\lambda_2 \lambda_1 v_2^p(x).\]
	Moreover, $\hat{u}_1(x)=\hat{u}_2(x)\equiv0$ on $A_4$.
\end{proposition}
\begin{proof}
	The first statement is obvious. The reward function in \eqref{eq:sup-couple1} is $h(x)-\lambda_1U^{(q+\lambda_1)}v^*(x, 2)$ for $x\in \II$.  Recall that for $x\in A_2$, $v^*(x, 2)=h(x)$ and for $x\in A_4$, $v^*(x, 2):=M_2(x)\neq h(x)$ whose form is determined later in this proof. For the function $M_2$, due to the continuity and nonnegativity of $v^*(x,2)$, we can take a continuous nonnegative function that coincides with $v^*(x, 2)$ on $A_4$. Consequently, the reward function for $v^*(x,1)$ becomes
	\[h(x)-\lambda_1U^{(q+\lambda_1)}h(x), \quad x\in A_2 \conn h(x)-\lambda_1U^{(q+\lambda_1)}M_2(x), \quad x\in A_4;\]
	therefore, it is not continuous on $\II$. Note that we need to choose $M_2(x)$ in a way that $\lambda_1U^{(q+\lambda_1)}M_2(x)<\infty,$ $x\in \II$ (see subsection \ref{subsec:ode-solution}). Similarly, we take a continuous nonnegative function $M_1(x)$ that coincides with $v^*(x,1)$ on $A_4$ and satisfies $\lambda_2U^{(q+\lambda_2)}M_1(x)<\infty, \;\; x\in \II$.
	
	Let $P_t^i$ denote the transition semigroup for $X^{(i)}$. For $x, y \in \II$ , we have
	\[
	\left|U^{(q+\lambda_1)}h(x)-U^{(q+\lambda_1)}h(y)\right|\le \int_0^\infty e^{-(q+\lambda_1)t}\left|\E_1^x[h(X_t^{(1)})]-\E_1^y[h(X_t^{(1)})]\right|\diff t=\int_0^\infty e^{-(q+\lambda_1)t}\left|P^1_th(x)-P^1_th(y)\right|\diff t\]
	and 
	\[
	\left|U^{(q+\lambda_1)}M_2(x)-U^{(q+\lambda_1)}M_2(y)\right|\le \int_0^\infty e^{-(q+\lambda_1)t}\left|\E_1^x[M_2(X_t^{(1)})]-\E_1^y[M_2(X_t^{(1)})]\right|\diff t=\int_0^\infty e^{-(q+\lambda_1)t}\left|P^1_tM_2(x)-P^1_tM_2(y)\right|\diff t,\]
	from which the continuity of the resolvents  follows due to the continuity of $h$ and $M_2$. Together with the assumed continuity of $h$, this shows the continuity of the reward functions $h-\lambda_1U^{(q+\lambda_1)}h$ and $h-\lambda_1U^{(q+\lambda_1)}M_2$ on $\II$.
	Furthermore, $U^{(q+\lambda_1)}h(x)$ satisfies the linear growth condition due to Assumption \ref{assump1}. This is because
	\begin{align*}\label{eq:resolvent_linearity}
	\left|U^{(q+\lambda_1)}h(x)\right|
	&\le \E_1^x\left[\int_0^{\infty}e^{-(q+\lambda_1)t}|h(X_t^{(1)})|\diff t\right]\le C\E_1^x\left[\int_0^{\infty}e^{-(q+\lambda_1)t}(1+|X_t^{(1)}|)\diff t\right]\\&\le C\int_0^{\infty}e^{-(q+\lambda_1)t}\left(1+\sqrt{\E_1^x\left[(X_t^{(1)})^2\right]}\right)\diff t
	\le C\int_0^{\infty}e^{-(q+\lambda_1)t}\left(1+\sqrt{De^{Dt}(1+x^2)}\right)\diff t\\&\le C\int_0^{\infty}e^{-(q+\lambda_1)t}\left(1+\sqrt{D}e^{\frac{D}{2}t}(1+|x|)\right)\diff t
	=\frac{C}{q+\lambda_1}+C\sqrt{D}(1+|x|)\frac{1}{q+\lambda_1-\frac{D}{2}}\\&\leq \left(C\sqrt{D}\frac{1}{q+\lambda_1-\frac{D}{2}}+\frac{C}{q+\lambda_1}\right)(1+|x|), \quad x\in\II
	\end{align*}
	where $C$ and $D$ are positive constants (see \citet[Theorem 2.9, Chapter 5]{karatzas-shreve-book1} and \citet[Lemma 3.1]{pham2007}). This guarantees that the reward function $h-\lambda_1U^{(q+\lambda_1)}h$ satisfies the linear growth condition. Since $M_1$ and $M_2$ are nonnegative functions, we see that the reward functions $h-\lambda_1U^{(q+\lambda_1)}M_2$ and $h-\lambda_2U^{(q+\lambda_2)}M_1$ also satisfy the linear growth condition. 

Now that the reward function in \eqref{eq:sup-couple} is known to satisfy the linear growth condition, we invoke Proposition \ref{prop:pham} to conclude that
$\sup_{\tau_1\in\S}\E_1^x[e^{-(q+\lambda_1)\tau_1}(h(X_{\tau_1}^{(1)})-\lambda_1U^{(q+\lambda_1)}v^*(X_{\tau_1}^{(1)},2))]$ solves $(q+\lambda_1-\G_1)u =0$ in $A_2$ and $A_4$.
	Similarly, we can show that   $\sup_{\tau_2\in\S}\E_2^x[e^{-(q+\lambda_2)\tau_2}(h(X_{\tau_2}^{(2)})-\lambda_2U^{(q+\lambda_2)}v^*(X_{\tau_2}^{(2)},1))]$ solves $(q+\lambda_2-\G_2)u=0$ in $A_3$ and $A_4$.
	
	Since $(q+\lambda_1-\G_1)\lambda_1U^{(q+\lambda_1)}v^*(x,2)=\lambda_1v^*(x,2)$, by applying $(q+\lambda_1-\G_1)$ on \eqref{eq:sup-couple1}, we obtain from Proposition \ref{prop:v-star-version}
	\[
	(q+\lambda_1-\G_1)v^*(x,1)=\lambda_1v^*(x,2) \quad x\in A_2, \: x\in A_4\]
	In $A_2$, $v^*(x,2)=h(x)$ and we have
	$(q+\lambda_1-\G_1)v^*(x,1)=\lambda_1h(x)$.
	Therefore, we obtain \[v^*(x,1)=u_1(x)+\lambda_1U^{q+\lambda_1}h(x)\quad x\in A_2\]
	with $u_1$ satisfying $(q+\lambda_1-\G_1)u_1(x)=0$ in $x\in A_2$.	
	Similarly, for $v^*(x,2)$ first we have
	$(q+\lambda_2-\G_2)v^*(x,2)=\lambda_2v^*(x,1)$ for  $x\in A_3, \: x\in A_4$.
	In $A_3$, $v^*(x,1)=h(x)$. Therefore, we obtain $v^*(x,2)=u_2(x)+\lambda_2U^{q+\lambda_2}h(x) \quad x\in A_3,$ where $u_2$ satisfies $(q+\lambda_2-\G_2)u_2(x)=0$ for $x\in A_3$. This completes the proof for regions $A_2$ and $A_3$.
	
	Finally, for  $x\in A_4$, the value functions can be found as solutions to the following paired optimal stopping problems:
	\[v^*(x,1)=\sup_{\tau_1\in\S}\E_1^x\left[e^{-(q+\lambda_1)\tau_1}\left(h(X_{\tau_1}^{(1)})-\lambda_1U^{(q+\lambda_1)}M_2(X_{\tau_1}^{(1)})\right)\right]+\lambda_1U^{(q+\lambda_1)}M_2(x)\]
	\[v^*(x,2)=\sup_{\tau_2\in\S}\E_2^x\left[e^{-(q+\lambda_2)\tau_2}\left(h(X_{\tau_2}^{(2)})-\lambda_2U^{(q+\lambda_2)}M_1(X_{\tau_2}^{(2)})\right)\right]+\lambda_2U^{(q+\lambda_2)}M_1(x)\]
	Let us set \[v_1^p(x):=\lambda_1U^{(q+\lambda_1)}M_2(x) \quad \text{and} \quad v_2^p(x):=\lambda_2U^{(q+\lambda_2)}M_1(x).\]
	We have for $i=1,2$
	\begin{equation}\label{eq:newly-added-vstar}
	v^*(x,i)=\hat{u}_i(x)+v_i^p(x), \quad x\in A_4.
	\end{equation}
	Furthermore,
	\begin{subequations}\label{eq:GRnew}
		\begin{equation}\label{eq:GR1new}
		(q+\lambda_1-\G_1)v_1^p(x)=\lambda_1 M_2(x)
		\end{equation}
		and
		\begin{equation}\label{eq:GR2new}
		(q+\lambda_2-\G_2)v_2^p(x)=\lambda_2 M_1(x).
		\end{equation}
	\end{subequations}
	Applying $(q+\lambda_2-\G_2)$ on \eqref{eq:GR1new}, we have for $x\in A_4$
	\begin{align*}
	(q+\lambda_2-\G_2) (q+\lambda_1-\G_1)v_1^p(x)&=\lambda_1(q+\lambda_2-\G_2)M_2(x)\\
	&=\lambda_1(q+\lambda_2-\G_2)(\hat{u}_2(x)+v_2^p(x))\\
	&=\lambda_1\lambda_2 M_1(x)
	\end{align*}
	where we used the definition of $\hat{u}_2(x)$ and \eqref{eq:GR2new} in the last equality. This implies, by the definition of $\hat{u}_1(x)$,
	\begin{align}\label{eq:ode11}
	(q+\lambda_2-\G_2) (q+\lambda_1-\G_1)v^*(x, 1)=\lambda_1 \lambda_2 M_1(x)=\lambda_1 \lambda_2 v^*(x, 1), \quad x\in A_4,
	\end{align}
	from which
	\begin{equation}\label{eq:v1-recursive-new}
	v_1^p(x)=\lambda_1\lambda_2U^{(q+\lambda_1)}U^{(q+\lambda_2)}M_1(x) \quad x\in A_4
	\end{equation}
	\emph{provided that the resolvents can be defined and are finite}.
	Similarly, we obtain
	\begin{align}\label{eq:ode22}
	(q+\lambda_1-\G_1) (q+\lambda_2-\G_2)v^*(x, 2)=\lambda_2 \lambda_1 v^*(x, 2), \quad x\in A_4.
	\end{align}
	Then in region $A_4$, we have
	\begin{align*}
	v^*(x, 1)&=\hat{u}_1(x)+\lambda_1U^{(q+\lambda_1)}M_2(x)\\
	&=\hat{u}_1(x)+\lambda_1U^{(q+\lambda_1)}\hat{u}_2(x)+\lambda_1U^{(q+\lambda_1)}v_2^p(x)\\
	&=\hat{u}_1(x)+\lambda_1U^{(q+\lambda_1)}\hat{u}_2(x)+\lambda_1\lambda_2U^{(q+\lambda_1)}U^{(q+\lambda_2)}M_1(x)\\
	&=\hat{u}_1(x)+\lambda_1U^{(q+\lambda_1)}\hat{u}_2(x)+v_1^p(x)
	\end{align*}
	where we used \eqref{eq:v1-recursive-new} in the last equality. For the final equality to hold, in view of \eqref{eq:newly-added-vstar} we need  $\lambda_1U^{(q+\lambda_1)}\hat{u}_2(x)=0$.  Since $\lambda_1>0$, it follows that $\hat{u}_2(x)=0$ for almost every $x$ in region $A_4$.  But the continuity of $v^*(x, 2)$ and hence of $\hat{u}_2(x)$ entails that $\hat{u}_2(x)=0$ for all $x$ in region $A_4$. The same argument also shows that $\hat{u}_1(x)=0$ for all $x$ in region $A_4$.
\end{proof}
We shall denote $v^*(x,i)$, $i=1,2$ in region $A_4$ by $v_i^p(x)$. In view of $\hat{u}_i(x)\equiv0$, $i=1,2$,
\begin{equation}\label{eq:coupled-vp}
\begin{cases}
v^*(x,1)=v^p_1(x)=\lambda_1U^{q+\lambda_1}v_2^p(x)\\
v^*(x,2)=v^p_2(x)=\lambda_2U^{q+\lambda_2}v_1^p(x)
\end{cases}
\end{equation}
and $v_i^p(x)$ solve \eqref{eq:ode11} and \eqref{eq:ode22}.

\subsection{On the solution of \eqref{eq:ode11}}\label{subsec:ode-solution}
Once we solve \eqref{eq:ode11} to find $v^p_1(x)$, we can have $v^p_2(x)$ immediately, provided that the resolvent $\lambda_2U^{(q+\lambda_2)}v_1^p(x)<\infty$. The finiteness of the resolvent should be one of the necessary conditions when finding the solution since the value functions are continuous and satisfy the linear growth condition on $\mathcal{I}$. Moreover, \eqref{eq:coupled-vp} becomes

\begin{align}\label{eq:couple-recursive-A4}
\begin{cases}
v_1^p(x)=\lambda_1 \lambda_2U^{(q+\lambda_1)}U^{(q+\lambda_2)}v_1^p(x),\\
v_2^p(x)=\lambda_1 \lambda_2U^{(q+\lambda_2)}U^{(q+\lambda_1)}v_2^p(x).
\end{cases}
\end{align}
To describe the solution to \eqref{eq:ode11}, let us take an example of the regime-switching geometric Brownian motion $X_t$, satifsying  $\diff X_t=\mu_{\eta_t}X_t\diff t+\sigma_{\eta_t}X_t\diff B_t$, for our state variable (see \citet{guo2004}). The solution to \eqref{eq:ode11} becomes $v_1^p(x)=\sum_{i=1}^4 k_i x^{\beta_i}$ and $\beta$'s  are the roots of the equation
\begin{equation}\label{eq:root}
j_1(\beta)j_2(\beta)=\lambda_1\lambda_2
\end{equation}
where \begin{align*}
\begin{cases}
j_1(\beta)=q+\lambda_1-\left(\mu_1-\frac{1}{2}\sigma_1^2\right)\beta-\frac{1}{2}\sigma_1^2\beta^2,\\
j_2(\beta)=q+\lambda_2-\left(\mu_2-\frac{1}{2}\sigma_2^2\right)\beta-\frac{1}{2}\sigma_2^2\beta^2.
\end{cases}
\end{align*}
While we have four distinct roots $\beta_1<\beta_2<0<\beta_3<\beta_4$,
we have a restriction: since $v_2^p(x)=\lambda_2U^{(q+\lambda_2)}v^p_1(x)$,
\begin{align*}
v_2^p(x)=\lambda_2\E^x_2\left[\int_0^\infty e^{-(q+\lambda_2)t}\left(k_1 (X_t^{(2)})^{\beta_1}+k_2(X_t^{(2)})^{\beta_2}+k_3(X_t^{(2)})^{\beta_3}+k_4(X_t^{(2)})^{\beta_4}\right)\diff t\right].
\end{align*}
Using   $X_t^{\beta_i}=x^{\beta_i}\cdot\exp\left(\beta_i(\mu_2-\frac{1}{2}\sigma_2^2)t+\beta_i\sigma_2 B_t\right)$, we have \[\E_2^x\left[(X_t^{(2)})^{\beta_i}\right]=x^{\beta_i}\cdot\exp\left(\beta_i(\mu_2-\frac{1}{2}\sigma_2^2)t+\frac{1}{2}\beta_i^2\sigma_2^2t\right)\]
and
\begin{align*}
v^p_2(x)=\lambda_2\sum\limits_{i=1}^{4}\left(k_ix^{\beta_i}\int_0^\infty e^{-j_2(\beta_i)t}\diff t\right).
\end{align*}
For this function to be finite, we need $j_2(\beta_i)>0$ for each $i$. If this condition is violated, then we need to set $k_i=0$ for $i$ such that $j_2(\beta_i)<0$. Above, we first found $v^p_1(x)$ and then derived necessary conditions $j_2(\beta_i)>0 \; \; \forall i$. If we had first found $v^p_2(x)$, the necessary conditions would become $j_1(\beta_i)>0 \; \; \forall i$. This means that any solution $\beta$ must satisfy both \[j_1(\beta)>0 \conn j_2(\beta)>0.\] As it will be proved below, there are only two solutions to $j_1(\beta)j_2(\beta)=\lambda_1\lambda_2$ that satisfy this condition, and they have opposite signs.
\begin{proposition}\label{prop:one-beta}
	For the geometric Brownian motion, $j_1(\beta)j_2(\beta)=\lambda_1\lambda_2$ from \eqref{eq:root} has only two solutions ($\beta$), such that $j_1(\beta)>0$ and $j_2(\beta)>0$. The same is true for any regime-switching diffusion satisfing $\diff X_t=\mu_{\eta_t}\diff t+\sigma_{\eta_t}\diff B_t$, where $B_t$ is a standard Brownian motion and $\mu_i\in\R$, $\sigma_i>0$ are arbitrary constants for $i=1,2$. Moreover, these two solutions have opposite signs.
\end{proposition}
\begin{proof}
	The proof uses the technique from Remark 2.1 in \citet{guo2001}. Let
	\[f(\beta)=j_1(\beta)j_2(\beta)-\lambda_1\lambda_2.\]	
	First, let us find a solution to $j_1(x)=q+\lambda_1-\left(\mu_1-\frac{1}{2}\sigma_1^2\right)x-\frac{1}{2}\sigma_1^2x^2=0$. This quadratic equation has two roots
	\[x_1^{+}=\frac{-(\mu_1-\frac{1}{2}\sigma_1^2)+\sqrt{(\mu_1-\frac{1}{2}\sigma_1^2)^2+2\sigma_1^2(\lambda_1+q)}}{\sigma_1^2}>0,\]
	\[x_1^{-}=\frac{-(\mu_1-\frac{1}{2}\sigma_1^2)-\sqrt{(\mu_1-\frac{1}{2}\sigma_1^2)^2+2\sigma_1^2(\lambda_1+q)}}{\sigma_1^2}<0.\]
	Moreover, $j_1(x)>0$ on $(x_1^-,x_1^+)$. Similarly, $j_2(x)=0$ has two solutions
	\[x_2^{+}=\frac{-(\mu_2-\frac{1}{2}\sigma_2^2)+\sqrt{(\mu_2-\frac{1}{2}\sigma_2^2)^2+2\sigma_2^2(\lambda_2+q)}}{\sigma_2^2}>0,\]
	\[x_2^{-}=\frac{-(\mu_2-\frac{1}{2}\sigma_2^2)-\sqrt{(\mu_2-\frac{1}{2}\sigma_2^2)^2+2\sigma_2^2(\lambda_2+q)}}{\sigma_2^2}<0,\]
	and $j_2(x)>0$ on $(x_2^-,x_2^+)$. Now, $f(x_1^-)=f(x_1^+)=f(x_2^-)=f(x_2^+)=-\lambda_1\lambda_2<0$. Also, we have $f(0)>0$, $f(\infty)>0$, and $f(-\infty)>0$. Figure \ref{fig:f} is an example graph of $f$. The blue dots on the horizontal axis indicate the values of $x_1^-,x_1^+,x_2^-,x_2^+$. Their magnitude relation does not affect the proof. The red area between the largest negative $x_{\cdot}^-$ and the smallest positive $x_{\cdot}^+$ is the area of $x$ where both $j_1(x)>0$ and $j_2(x)>0$ hold. Based on the abovementioned characteristics of $f$ and on the fact that $f$ will cross the horizontal axis only four times, we know that $f$ crosses the horizontal axis only twice on the red area, resulting in two solutions of opposite signs. The proof for the Brownian motion with drift is similar, with $j_1$ and $j_2$ adjusted accordingly using \eqref{eq:ode11}.
	\begin{figure}[h]
		\includegraphics[scale=0.8]{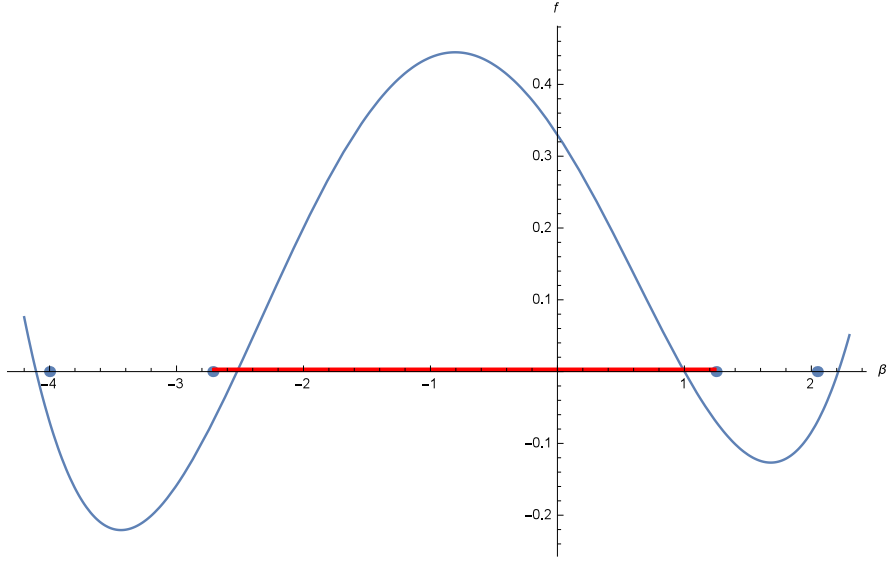}
		\caption{An example graph of $f$ in the proof of Proposition \ref{prop:one-beta}}\label{fig:f}
	\end{figure}
\end{proof}
\indent Since the two appropriate solutions from Proposition \ref{prop:one-beta} have opposite signs, it is intuitive that one of them may be eliminated, depending on the form of $A_4$, due to the boundedness condition of value functions on $A_4$ (e.g., see \citet{guo2004} where positive roots are eliminated). The boundedness condition is proved in Proposition \ref{prop:one_solution} below. Note that when handling the Ornstein-Uhlenbeck process or its exponential version, the fundamental solutions $\varphi_{q}$ and $\psi_{q}$ are special functions: the parabolic cylinder function and the confluent hypergeometric function of the second kind, respectively (see \citet[Sects. 10.2 and 9.9]{lebedev}) and therefore it is not easy to compute how many roots of $\beta$ satisfy $j_1(\beta)>0$ and $j_2(\beta)>0$ and how many of them should be eliminated. The next proposition overcomes this technical difficulty by first showing that the value functions must be bounded on $A_4$ for any underlying diffusion $X$. It then extends the result of Proposition \ref{prop:one-beta} ensuring that the solution to \eqref{eq:couple-recursive-A4}, in the case where one of the endpoints of $A_4$ is a natural boundary, is a  multiple of only one function; thus, only one weight $k$ needs to be determined.

\begin{proposition}\label{prop:one_solution} The value functions $v_1^p(x)$ and $v_2^p(x)$ in region $A_4=\mathrm{C}_1 \cap \mathrm{C}_2$ are  bounded on $A_4$.  Moreover, when one of the endpoints of the interval $A_4$ is the natural boundary,
	the solution to \eqref{eq:couple-recursive-A4} is $v_1^p(x)=kz(x)$ ($v_2^p$ will be found simultaneously) for some continuous function $z$ bounded in $A_4$. This holds for any underlying diffusion $X$ whose parameters $\mu(\cdot,i)$ and $\sigma(\cdot,i)$ satisfy \eqref{eq:Lipschitz-and-growth} for each $i=1,2$ and only one weight $k$ needs to be determined.
\end{proposition}
\begin{proof} We first show that $\sup_{x\in A_4}v_i^p(x)<\infty$, $i=1,2$ by using a contradiction. Let us assume that for stopping region $\Gamma_i$, $h(x_0)<\sup_{x\in A_4}v_i^p(x)$ for $\forall x_0\in\Gamma_i$. This implies the nonexistence of the stopping region $\Gamma_i$ and we exclude such case. Hence if there is a stopping region $\Gamma_i$,  there exists $x_0\in \Gamma_i$ such that  $\sup_{x\in A_4}v_i^p(x)\le h(x_0)<\infty$.
	
	Let us assume that one of the endpoints of $A_4$ is the left boundary $\ell$. The proof is similar when one of the endpoints of $A_4$ is $r$. It is natural to define $C^*(\ell,b]$, a set of \emph{bounded} continuous functions on $(\ell,b]$ for some $b<r$. We know that $C^*(\ell,b]$ is a complete metric space with respect to the sup norm
	\[||f-g||:=\sup\limits_{x\in (\ell,b]}|f(x)-g(x)|\]
	(see \citet[pp.~150-151]{rudin_math_analysis}).
	We have
	\[\lambda_1\lambda_2U^{(q+\lambda_1)}U^{(q+\lambda_2)}f(x)=\E_1^x\left[e^{-qT^1}\E_2^{X_{T^1}}\left[e^{-qT^2}f(X_{T^2})\right]\right] \quad f\in C^*(\ell,b]\]
	Furthermore, $\E_1^{\cdot}\left[e^{-\alpha T^1}\right]=\int_0^{\infty}e^{-\alpha t}\lambda_1e^{-\lambda_1t}\diff t=\frac{\lambda_1}{\alpha+\lambda_1}$ for $\alpha>0$. Then, $\lim\limits_{\alpha\downarrow 0}\E_1^{\cdot}\left[e^{-\alpha T^1}\right]=1$ and $T^1<\infty$ almost surely. Similarly, $T^2<\infty$ holds almost surely. By taking sufficiently large $b$, we can ensure that
	$\E_2^{X_{T^1}}\left[e^{-qT^2}f(X_{T^2})\right]\le ||f||\cdot \E_2^{X_{T^1}}\left[e^{-qT^2}\right]$ almost surely and  $A_4\subset(\ell,b]$.
	Then, for $f,g\in C^*(\ell,b]$, we have
	\begin{align*}
	\left|\left|\lambda_1\lambda_2U^{(q+\lambda_1)}U^{(q+\lambda_2)}f-\lambda_1\lambda_2U^{(q+\lambda_1)}U^{(q+\lambda_2)}g\right|\right|&\le\E_1^x\left[e^{-qT^1}\left|\left|f-g\right|\right|\E_2^{X_{T^1}}\left[e^{-qT^2}\right]\right]=\left|\left|f-g\right|\right|\E_1^x\left[e^{-qT^1}\frac{\lambda_2}{q+\lambda_2}\right]\\
	&=\left|\left|f-g\right|\right|\frac{\lambda_1}{q+\lambda_1}\frac{\lambda_2}{q+\lambda_2} \qquad x\in(\ell,b]
	\end{align*}
	Since $\frac{\lambda_1}{q+\lambda_1}\frac{\lambda_2}{q+\lambda_2}<1$, it follows that $\lambda_1\lambda_2U^{(q+\lambda_1)}U^{(q+\lambda_2)}$is a contraction map on $C^*(\ell, b]$. Then, due to the contraction mapping theorem (see \citet[Section 10.3]{royden_real_analysis}), there is only one fixed point $f\in C^*(\ell,b]$ satisfying $\lambda_1\lambda_2U^{(q+\lambda_1)}U^{(q+\lambda_2)}f=f$ . In particular, we have only one solution to the first equation in \eqref{eq:couple-recursive-A4}.
\end{proof} 
\bigskip

\section{Method to Find Optimal Solution}\label{sec:method}
In Section \ref{sec:formulation}, we characterized the value functions of the problem \eqref{eq:problem} in each region of $A_j, j=1, 2, 3, 4$ being defined as
\[\II=(\Gamma_1 \cap \Gamma_2) \cup (\mathrm{C}_1 \cap \Gamma_2)\cup (\mathrm{C}_2 \cap \Gamma_1)\cup (\mathrm{C}_1\cap \mathrm{C}_2):=A_1 \cup A_2 \cup A_3\cup A_4.
\]
Based on this characterization, we will show how to obtain the value functions. Our characterization of the value functions in Section 3 (Proposition \ref{prop:regions}) does not presume any specific order or form of $A_j$, $j=1,2,3,4$. Therefore, in case of a complex structure, for example, disconnected stopping regions, every pattern of partitioning the state space $\II$ can still be written as combinations of $A_j$. 
Thus, the necessary and sufficient conditions discussed below, as well as the procedure of estimation, are applicable to all the patterns.

\emph{Since we know, from Proposition \ref{prop:regions}, the form of $v^*(x, i), i=1, 2$ in each region, we can apply the result of $F$-concavity characterization of  \citet[Propositions 5.11 and 5.12]{DK2003} with the aid of Proposition \ref{prop:pham} in each region}. \\

\noindent \underline{Transformation :}
Let us define as in \eqref{eq:F}
\[
F_1(x):=\frac{\psi_{q+\lambda_1}(x)}{\varphi_{q+\lambda_1}(x)} \conn F_2(x):=\frac{\psi_{q+\lambda_2}(x)}{\varphi_{q+\lambda_2}(x)}
\]
By the $F$-concavity characterization of the value function which we briefly reviewed in Section \ref{sec:math-prep}, for each $i=1,2$, we need to find the smallest nonnegative concave majorant, denoted by $W_i(F_i(x))$, of the appropriate reward function in the transformed space by $F_i$.  Recall \eqref{eq:transform} for this transformation. At the same time, we need to make sure that the necessary conditions in Propositions 4.1 below are satisfied. \\
\indent Let us stress first that this method differs from so-called ``guess and verify" method which requires some difficult algebraic computations that are usually problem-specific. In contrast, our method is to find solutions systematically using the geometric argument of concavity.\\
\indent We analyze the threshold points between regions $A_j$, summarized in Table \ref{tbl:thresholds}.
The table excludes the thresholds between $A_2/A_3$ as well as $A_1/A_4$ regions because these thresholds do not exist in our setting as we show below. 
\begin{table}[h]
		\caption{Summary of the threshold points and the two regions that each point separates. The threshold for $A_2/A_3$ and $A_1/A_4$ do not exist in our setting.} 
		\begin{tabular}{cccccc}
			Threshold & Separated Regions & & & Threshold & Separated Regions\\
			\cline{1-2} \cline{5-6}
			$a$ & $A_1$/$A_2$  & & & $b$ & $A_1$/$A_3$\\
			$c$ & $A_2$/$A_4$ & & & $d$ & $A_3$/$A_4$
		\end{tabular}\label{tbl:thresholds}
\end{table}
\indent First, suppose that $A_2$ and $A_3$ are adjacent regions. Since for $X^{(1)}$, $A_2$ is a continuation region and $A_3$ is a stopping region, there exist $y\in A_2$ and $z\in A_3$ such that $h(y)<h(z)$. This in turn contradicts the fact that $A_2$ is a stopping region and $A_3$ is a continuation region for $X^{(2)}$ since the reward $h$ is the same. Next, assume that $A_4$ is adjacent to $A_1$ region and set the threshold point as $s$. Without loss of generality, we have $v_1^p(x)-h(x)>v_2^p(x)-h(x)>0$ for some $x\in A_4$ in the neighborhood of $s$ (with $x<s$). This indicates that the value obtained by waiting is higher for $X^{(1)}$. In order to have $v_1^p(s)-h(s)=v_2^p(s)-h(s)=0$, the relationship must reverse and $v_1^p(y)-h(y)<v_2^p(y)-h(y)$ must hold for some $x<y\le s$. This cannot happen since all other settings are unchanged. We summarize these findings as a Remark for future reference.\\
\begin{remark}\label{rmrk:adjacent}
	\begin{normalfont}
    The following two statements hold true:
    \begin{enumerate}
    	\item The region $A_1$ cannot be adjacent to the region $A_4$. 
    	\item The region $A_2$ cannot be adjacent to the region $A_3$.
    \end{enumerate}
	\end{normalfont}
\end{remark}
\indent By symmetry, the thresholds $a$ and $b$ as well as $c$ and $d$ have the same characterization; therefore without loss of generality, we will discuss only $a$ and $c$. Due to the interacting nature of \eqref{eq:sup-couple}, we need to find $W_1(F_1(x))$ and $W_2(F_2(x))$, together with $a$ and $c$, simultaneously in their respective transformed spaces. The regions $A_j$ may be disconnected, resulting in multiple $a$ and $c$ thresholds;  however, the necessary conditions of Proposition \ref{prop:necessary_cond} are the same for any such $a$ and $c$ regardless of their number. In order to analyze the thresholds, we construct the following Tables \ref{tbl:rewards1} and \ref{tbl:rewards2}:
\begin{table}[h]
\caption{The reward function in \eqref{eq:sup-couple1} and the value function  $v^*(x, 1)$. Note that $v^p_1(\cdot)$ in the last row is the solution to \eqref{eq:ode11}}
		\begin{tabular}{|c|c|c|c|}
			\hline
			Region & $A_1$ & $A_2$ & $A_4$ \\
			\hline
			Reward function & $\underset{(H_{11})}{h-\lambda_1U^{(q+\lambda_1)}h}$  &  $\underset{(H_{11})}{h-\lambda_1U^{(q+\lambda_1)}h}$  &  $\underset{(H_{12})}{h-v^p_1}$\\
			\hline
			Value function & $h$ & $u_1+\lambda_1U^{(q+\lambda_1)}h$ & $v^p_1$ \\
			\hline
		\end{tabular}\label{tbl:rewards1}
\end{table}
\newline Note that  $H_{1 j}, j=1, 2$ in the parentheses (Table \ref{tbl:rewards1}) denote the names of the reward functions transformed by $F_1(x)$.
\begin{align}\label{eq:H1-all}
\begin{cases}
H_{11}(y)=\frac{h-\lambda_1 U^{(q+\lambda_1)}h}{\varphi_{q+\lambda_1}} \circ F_1^{-1}(y), & y\in F_1(A_1)\cup F_1(A_2), \\
H_{12}(y)=\frac{h-\lambda_1 U^{(q+\lambda_1)}v^p_2}{\varphi_{q+\lambda_1}} \circ F_1^{-1}(y), &y\in F_1(A_4).
\end{cases}
\end{align}		
	
\begin{table}[h]
\caption{The reward function in \eqref{eq:sup-couple2} and the value function  $v^*(x, 2)$}
		\small{\begin{tabular}{|c|c|c|c|}
				\hline
				Region & $A_1$ & $A_2$ & $A_4$ \\
				\hline
				Reward function & $\underset{(H_{21})}{h-\lambda_2U^{(q+\lambda_2)}h}$ & $\underset{(H_{22})}{h-\lambda_2U^{(q+\lambda_2)}[u_1+\lambda_1U^{(q+\lambda_1)}h]}$ & $\underset{(H_{23})}{h-v^p_2}$ \\
				\hline
				Value function & $h$ & $h$  & $v^p_2=\lambda_2U^{(q+\lambda_2)}v^p_1$ \\
				\hline
		\end{tabular}}\label{tbl:rewards2}
	\end{table}
	Note also  that  $H_{2 j}, j=1, 2, 3$ in the parentheses (Table \ref{tbl:rewards2}) denote the names of the reward functions transformed by $F_2(x)$.
	\begin{align}\label{eq:H2-all}
	\begin{cases}
	H_{21}(y)=\frac{h-\lambda_2 U^{(q+\lambda_2)}h}{\varphi_{q+\lambda_2}} \circ F_2^{-1}(y), \quad y\in F_2(A_1),\\
	H_{22}(y)=\frac{h-\lambda_2U^{(q+\lambda_2)}[u_1+\lambda_1U^{(q+\lambda_1)}h] }{\varphi_{q+\lambda_2}} \circ F_2^{-1}(y), \quad y\in F_2(A_2),\\
	H_{23}(y)=  \frac{h-v^p_2}{\varphi_{q+\lambda_2}} \circ F_2^{-1}(y), \quad y\in F_2(A_4).
	\end{cases}
	\end{align}
The following proposition is based on Proposition \ref{prop:regions}.
\begin{proposition}\label{prop:necessary_cond}
		Necessary Conditions.\\
\noindent \begin{normalfont}(N-1)	\end{normalfont} The continuity of the value functions (Lemma \ref{lem:hitting}) provides two conditions at $c$: 
\begin{align}\label{eq:condition1}
u_1(c)+\lambda_1U^{(q+\lambda_1)}h(c)=v^p_1(c),
\end{align}
\begin{align}\label{eq:condition2}
h(c)=v^p_2(c)=\lambda_2U^{(q+\lambda_2)}v^p_1(c).
\end{align}
If $h$ is differentiable everywhere on $\{x\in\mathcal{I}: h(x)>0\}$, the following condition also needs to be satisfied for optimality:
\begin{equation}\label{eq:condition3}
(h(c)-v^p_2(c))'=h'(c)-\left(\lambda_2U^{(q+\lambda_2)}v^p_1(x)\right)'\bw_{x=c}=0.
\end{equation}
Conditions \eqref{eq:condition2} and \eqref{eq:condition3} are equivalent to
\begin{align}\label{eq:1-2}
\begin{cases}
H_{23}(F_2(c))=0,\\
H'_{23}(F_2(c))=0.
\end{cases}
\end{align}
Whether $h$ is differentiable or not, we must choose $c$ such that the maximum of $H_{23}(y)$ on $F_2(A_4)$ is attained at $y=F_2(c)$ and $H_{23}(F_2(c))=0$.\\
\begin{normalfont}(N-2)	\end{normalfont} 
In $A_2$ region, $u_1(x)=\varphi_{q+\lambda_1}(x)W(F_1(x))$ where $W(F_1(x))=A F_1(x)+B$ is the smallest nonnegative concave majorant of $H_{11}$ in $F_1(A_1)\cup F_1(A_2)$ with some $A, B\in \R$ to be determined. The threshold $a$ is identified as a point where $W(F_1(a))=H_{11}(F_1(a))$. Due to the continuity of the value function, the following condition holds at $a$:
\begin{equation}\label{eq:condition-a}
u_1(a)=h(a)-\lambda_1U^{(q+\lambda_1)}h(a).
\end{equation}
If $h$ is differentiable everywhere on $\{x\in\mathcal{I}: h(x)>0\}$, the following must also be satisfied for the optimality: 
\begin{equation}\label{eq:condition-a2}
u_1'(a)-\left(h'(a)-\left(\lambda_1U^{(q+\lambda_1)}h(x)\right)'\bw_{x=a}\right)=0.
\end{equation}
The conditions \eqref{eq:condition3} and \eqref{eq:condition-a2} coincide with the celebrated smooth-fit condition. In both differentiable and nondifferentiable cases, $c$ and $a$ are obtained by finding the smallest nonnegative concave majorants of the transformed rewards $H_{23}$ and $H_{11}$, respectively. 
\end{proposition}
\begin{proof}
Equations \eqref{eq:condition1}, \eqref{eq:condition2} and \eqref{eq:condition-a} are immediate consequences of Lemma \ref{lem:hitting} and Proposition \ref{prop:regions}. For \eqref{eq:condition3}, let us consider the function $M(x):=h(x)-\lambda_2U^{(q+\lambda_2)}v^p_1(x)$ (see \eqref{eq:g21}). The first term is the reward that one obtains when stopping immediately at $x$.  The second term is the discounted value that one can expect when not stopping at $x$.  It can be regarded as the opportunity cost of stopping  and hence the first order condition for the optimality is $M'(x)=0$. For the reward function $H_{2\cdot}(y)$ in the transformed space, we have
\[{\frac{\diff H_{2\cdot}(y)}{\diff y}=\left[\frac{1}{F_2'}\left(\frac{(h-v^p_2)'\varphi_{q+\lambda_2}-(h-v^p_2)\varphi'_{q+\lambda_2}}{\varphi^2_{q+\lambda_2}}\right)\right]\circ F_2^{-1}(y).}\]
Setting $y=F_2(c)$, the conditions \eqref{eq:condition2} and \eqref{eq:condition3} imply \eqref{eq:1-2}. \\
\indent The $u_1$ function solves \eqref{eq:qvi-1} and in view of \eqref{eq:v=phiW},
\begin{equation}\label{eq:u1}
u_1(x)=\varphi_{q+\lambda_1}(x)W(F_1(x))=
A\psi_{q+\lambda_1}(x)+B\varphi_{q+\lambda_1}(x)
\end{equation} with some $A, B\in \R$ to be determined. Hence, $W_1(F_1(x))$ is in the form of
\[
W_1(F_1(x))=A F_1(x)+B
\] on $F_1(A_2)$ and majorizes the reward function $H_{11}$. 
Since point $a$ is the threshold between $A_1$ and $A_2$ regions and is in the stopping region of $v^*(x,2)$, the reward function for $v^*(x,1)$ is $h(x)-\lambda_1U^{(q+\lambda_1)}h(x)$ in both $A_1$ and $A_2$ regions. Hence we can identify the boundary point $a$ by finding the smallest nonnegative concave majorant of $h(x)-\lambda_1U^{(q+\lambda_1)}h(x)$ transformed by $F_1(x)$ in $F_1(A_2)\cup F_1(A_1)$. The condition \eqref{eq:condition-a2} is derived similarly to \eqref{eq:condition3}.
\end{proof}

The following proposition derives the geometric conditions that are sufficient for optimality. This is just to confirm that the value functions are smallest nonnegative concave majorants of respective rewards in each region under necessary conditions.

\begin{proposition}\label{prop:sufficient-cond}
Sufficient Conditions. \\
Given that the necessary conditions in Proposition \ref{prop:necessary_cond} are satisfied, we have the following four sufficient conditions for the optimality of the value functions:
\begin{enumerate}
\item [(S-1)] $h(x)-\lambda_1U^{(q+\lambda_1)}v_2^p(x)=h(x)-v^p_1(x)$ transformed by $F_1(x)$ is majoraized by the horizontal axis in $F_1(A_4)$.
\item [(S-2)]$h(x)-\lambda_2U^{(q+\lambda_2)}v_1^p(x)$ transformed by $F_2(x)$ is majoraized by the horizontal axis in $F_2(A_4)$.
\item [(S-3)] $u_1(x)$ transformed by $F_1(x)$ is the smallest nonnegative concave majorant of $h(x)-\lambda_1U^{(q+\lambda_1)}h(x)$ transformed by $F_1(x)$ in $F_1(A_2)$ and $h(x)-\lambda_1U^{(q+\lambda_1)}h(x)$ transformed by $F_1(x)$ is concave in $F_1(A_1)$ .
\item [(S-4)] $h(x)-\lambda_2U^{(q+\lambda_2)}h(x)$ transformed by $F_2(x)$ is concave in $F_2(A_1)$ and $h(x)-\lambda_2U^{(q+\lambda_2)}[u_1+\lambda_1U^{(q+\lambda_1)}h](x)$ transformed by $F_2(x)$ is concave in $F_2(A_2)$.
\end{enumerate}
\end{proposition}
\begin{proof}
It is proved in Proposition \ref{prop:regions} that in $A_4$ region, $\hat{u}_2(x)=\varphi_{q+\lambda_2}(x)W_2(F_2(x))\equiv 0$. Then the reward function $h-\lambda_2U^{(q+\lambda_2)}v_1^p(x)$ must be majorized by the horizontal axis in the transformed space by $F_2(x)$. The same is true for the reward function $h-\lambda_1U^{(q+\lambda_1)}v_2^p(x)$ since $\hat{u}_1(x)=0$ in $A_4$. The statement (S-3) holds by the definition of $u_1$. The final statement (S-4) is obvious by noting that $X^{(2)}$ is in its stopping region in $A_1$ and $A_2$.
\end{proof}

We summarize the procedure for finding the value functions.\\
\noindent\underline{Procedure :}
\begin{enumerate}
\item Given the diffusion, its parameters, and $\lambda_i, i=1, 2$, compute $(\psi_{q+\lambda_1}, \varphi_{q+\lambda_1})$, $(\psi_{q+\lambda_2}, \varphi_{q+\lambda_2})$, and $F_1, F_2$.
\item Solve \eqref{eq:ode11} for $v_1^p(x)$ which leads to $v_2^p(x)$ via \eqref{eq:couple-recursive-A4}.
\item Identify the threshold points in Table \ref{tbl:thresholds} that satisfy the necessary conditions of Proposition \ref{prop:necessary_cond}.
\item Make sure the sufficient conditions (S-1)$\sim$(S-4) in Proposition \ref{prop:sufficient-cond} are satisfied. 
\end{enumerate}
Once one obtains $v^*(x, 1)$ and $v^*(x, 2)$ that satisfy the conditions in Propositions \ref{prop:regions} and \ref{prop:necessary_cond} using the method we described, the pair is a solution if conditions (S-1)$\sim$ (S-4) in Proposition \ref{prop:sufficient-cond} are satisfied. No proof of the verification lemma  (\citet[Theorem 10.4.1]{oksendal-book})  for optimality is necessary. We simply need to check conditions (S-1)$\sim$ (S-4) geometrically. To illustrate this, we will solve a real-life problem in the next section.

\section{Example}\label{sec:example}
We study an American capped call option on dividend-paying stock (\citet{broadie-detemple1995} and \citet{DK2003}). The stock price is driven by
\[\diff X_t=(r-\delta_{\eta_t})X_t\diff t + \sigma_{\eta_t}X_t\diff B_t, \quad \eta_t \in \{1, 2\} \]
with constants $\sigma_{\eta_t}>0, r>0$, and $\delta_{\eta_t}>0$. The constants $r$ and $\delta$ indicate the risk-free interest rate and dividend rate, respectively. This process is a regime-switching geometric Brownian motion with state space $\II=(0, \infty)$ and both boundaries are natural.
The optimal stopping problem for the perpetual American call option with strike price $K\ge 0$ and the cap $L>K$ in regime-switching environment is
\[
v^*(x, i)=\sup_{\tau\in \S}\E_i^x[e^{-r\tau}(X_\tau \wedge L-K)^+]
\] under a risk-neutral probability measure with
\begin{align*}
h(x)=\begin{cases}
(x-K)^+, & 0<x\le L, \\
(L-K)^+, &x> L.
\end{cases}
\end{align*}
We can verify that the quantities in \eqref{eq:finiteness} are finite. It is natural to assume that the drift term $r-\delta_i>0$ for each $i$. \\ 
\noindent \emph{Step (1)}: Using the infinitesimal generator $\G_i:=\frac{\sigma^2_ix^2}{2}\frac{\diff^2 }{\diff x^2}+(r-\delta_i)x\frac{\diff }{\diff x}$, the fundamental solutions of $(r+\lambda_i)v(x)-\G_i v(x)=0$ are
\[\varphi_{r+\lambda_i}(x)=x^{\gamma_{i, 1}}\conn \psi_{r+\lambda_i}(x)=x^{\gamma_{i, 2}}, \quad i=1, 2
\] where $\gamma_{i, 1}<0$ and $\gamma_{i, 2}>0$ for $i=1, 2$.  These are the roots of the quadratic equation $(1/2)\sigma_i^2x^2+(r-\delta_i-(\sigma_i^2/2))x-(r+\lambda_i)=0$.  Define $F_i$ as in \eqref{eq:F}:
\[F_i(x)=x^{\theta_i}, \quad \theta_i=\gamma_{i, 2}-\gamma_{i, 1}>0, \quad x>0, \quad i=1, 2.\]

\noindent \emph{Step (2)}: By solving \eqref{eq:ode11}, the value function $v^*(x, 1)$ in region $A_4$ is of the form
\[v^*(x, 1)=k x^{\beta^*}
\] where $\beta^*$ is the smaller of the two positive roots of \eqref{eq:root} with
\[
j_i(\beta)=r+\lambda_i-\left(r-\delta_i-\frac{1}{2}\sigma_i^2\right)\beta-\frac{1}{2}\sigma_i^2\beta^2, \quad i=1, 2,
\] and $k$ is a constant to be determined. We discard (1) two negative roots since the value function must be bounded when $x\rightarrow 0$ and (2) the larger of the positive two roots due to the finiteness condition of the resolvent of $v^*(x, 1)$.\\
Note that $\beta^*>1$. Indeed, let us set
\[
f(\beta):=j_1(\beta)j_2(\beta)-\lambda_1\lambda_2
\] where  $f(\beta^*)=0$ holds.  Since $\delta_i, \lambda_i>0$, we have $f(1)=(\lambda_1+\delta_1)(\lambda_2+\delta_2)-\lambda_1\lambda_2>0$ and $f'(1)=-(\lambda_1+\delta_1)(r-\delta_2+\frac{1}{2}\sigma_2^2)-(\lambda_2+\delta_2)(r-\delta_1+\frac{1}{2}\sigma_1^2)<0$ since $r-\delta_i>0$ for $i=1,2$. The continuity of $f$ implies that $\beta^*>1$.\\
\noindent\emph{Step (3)}: This step is about identifying the threshold $c$ between $A_4/A_2$ regions. From the form of the reward $h$, it is obvious that the left boundary of $A_4$ region is $0$. Using Proposition \ref{prop:necessary_cond} and setting $D:=\frac{\lambda_2}{j_2(\beta^*)}$, we solve \eqref{eq:1-2} with
\begin{equation*}
H_{23}(y)=  \frac{h-\lambda_2U^{(r+\lambda_2)}kx^{\beta^*}}{\varphi_{r+\lambda_2}} \circ F_2^{-1}(y)=
\begin{cases}
-Dky^{\frac{\beta^*}{\theta_2}}y^{-\frac{\gamma_{2, 1}}{\theta_2}}, & y \in (0, F_2(K)],\\
\left(y^{\frac{1}{\theta_2}}-K-Dky^{\frac{\beta^*}{\theta_2}}\right)y^{-\frac{\gamma_{2, 1}}{\theta_2}}, & y \in (F_2(K), F_2(L)],\\
\left(L-K-Dky^{\frac{\beta^*}{\theta_2}}\right)y^{-\frac{\gamma_{2, 1}}{\theta_2}}, & y \in (F_2(L), \infty).
\end{cases}
\end{equation*}
The two conditions in \eqref{eq:1-2} provide \begin{equation}\label{eq:k-and-c}
k=\frac{c-K}{D}c^{-\beta^*} >0 \conn c=\frac{\beta^*K}{\beta^*-1}>K
\end{equation}
since $\beta^*>1$.
From the second equation, we have
\begin{equation}\label{eq:c}
c<L \Leftrightarrow \beta^*>\frac{L}{L-K}.
\end{equation}

Suppose that $c<L$ holds.  We wish to show that $F_2(c)$ is the only point that satisfies \eqref{eq:1-2}. Since $D>0$ and $k>0$, we see  that $H_{23}(F_2(x))<0$ and is decreasing on $x\in (0, K)$. We have $H_{23}(F_2(K))<0$ and it suffices to confirm that $H_{23}'(y)$ vanishes at most two points. The equation $H'_{23}(y)=0$ reduces to
\[
(1-\gamma_{2,1})y^{\frac{1}{\theta_2}}-Dk(\beta^*-\gamma_{2,1})y^{\frac{\beta^*}{\theta_2}}=-\gamma_{2,1}K>0.
\]  The right-hand side is a constant and it is easily seen that the left-hand side has only one extreme point, which proves that $H'_{23}(y)$ vanishes at most twice.

It would be beneficial that we show graphs in each step.  The parameter set here is $(r, \delta_1, \delta_2, \lambda_1, \lambda_2, \sigma_1, \sigma_2)=(0.3, 0.2, 0.225, 1, 1, 0.5, 0.3)$ with $K=5$ and $L=15$.  We find  $\beta^*=1.85235$ as well as $(\gamma_{1, 1}, \gamma_{1, 2})=(-3.1265, 3.3265)$ and $(\gamma_{2, 1}, \gamma_{2, 2})=(-5.7185, 5.0518)$. From \eqref{eq:k-and-c} we obtain $c=10.8661<L$ and $k=0.0770$. Figure \ref{fig:H23} shows that $H_{23}(y)$ is majorized by the horizontal axis and is tangent to it at $c$. Thus, the smallest nonnegative concave majorant $W_2(y)\equiv 0$ and condition (S-2) in Proposition \ref{prop:sufficient-cond} is confirmed to be satisfied.
\begin{figure}[htbp]
	\begin{center}
		\begin{minipage}{0.9\textwidth}
			\centering{\includegraphics[scale=0.8]{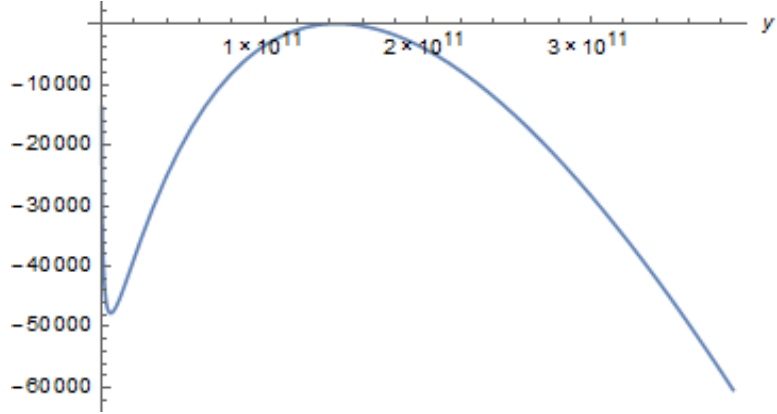}}\\
		\end{minipage}
		\caption{\small The graph of $H_{23}(y)$ on $[F_2(K), \infty)$ when $L=15$. } \label{fig:H23}
	\end{center}
\end{figure}

When choosing $c$ as a solution to \eqref{eq:k-and-c}, we may encounter two types of issues. First, the resulting $c$ may be greater than $L$ since \eqref{eq:k-and-c} does not take the constraint $L$ into consideration. Second, when setting $c$ as a solution to \eqref{eq:k-and-c}, the resulting value functions may be infeasible. The feasibility means that for some $x_0\in\mathcal{I}$, $v^*(x,i)\le h(x_0)$ must hold for all $x\in A_4$, $i=1,2$. We do not encounter these issues when $L=15$. We set $L=11.3$ in Section \ref{subsec:binding} and demonstrate the solution method for the case when the solution to \eqref{eq:k-and-c} leads to infeasible value functions. The same method can be used when the solution to \eqref{eq:k-and-c} is greater than $L$. 

Let us now come back to the original parameter set where $L=15$. We have already identified the threshold $c$ when $L=15$ and proved that there is only one point $c$ satisfying \eqref{eq:1-2}. From Remark \ref{rmrk:adjacent}, $A_4$ is adjacent to $A_2$. In the next step we will find the associated point $a$ that separates $A_2$ and $A_1$ regions. Let us prepare by setting up Tables \ref{tbl:example-rewards1} and \ref{tbl:example-rewards2} similar to the ones in Section \ref{sec:method}.

\begin{table}[h]
\caption{The reward function in \eqref{eq:sup-couple1} and the value function  $v^*(x, 1)$}
	\begin{tabular}{|c|c|c|c|}
		\hline
		Region & $A_4: (0, c)$  & $A_2$ & $A_1$ \\
		\hline
		Reward function & $\underset{(H_{12})}{h-kx^{\beta^*}}$  &  $\underset{(H_{11})}{h-\lambda_1U^{(r+\lambda_1)}h}$  &  $\underset{(H_{11})}{h-\lambda_1U^{(r+\lambda_1)}h}$ \\
		\hline
		Value function & $\quad kx^{\beta^*}\quad $ & $u_1+\lambda_1U^{(r+\lambda_1)}h$  & $h$ \\
		\hline
	\end{tabular}\label{tbl:example-rewards1}
\end{table}

\begin{table}[h]
\caption{The reward function in \eqref{eq:sup-couple2} and the value function  $v^*(x, 2)$}
	\small{\begin{tabular}{|c|c|c|c|}
			\hline
			Region & $A_4: (0, c)$ & $A_2$ & $A_1$ \\
			\hline
			Reward function & $ \underset{(H_{23})}{h-\lambda_2U^{(r+\lambda_2)}kx^{\beta^*}}$ & $\underset{(H_{22})}{h-\lambda_2U^{(r+\lambda_2)}[u_1+\lambda_1U^{(r+\lambda_1)}h]}$ & $\underset{(H_{21})}{h-\lambda_2U^{(r+\lambda_2)}h}$ \\
			\hline
			Value function & $\lambda_2U^{(r+\lambda_2)}kx^{\beta^*}$ & $h$  & $h$ \\
			\hline
	\end{tabular}}\label{tbl:example-rewards2}
\end{table}
Let us frist write down:
\begin{subequations}\label{eq:data}
	\begin{equation}\label{eq:Uv1-example}
	\lambda_2U^{(r+\lambda_2)}kx^{\beta^*}=\frac{\lambda_2}{j_2(\beta^*)}kx^{\beta^*}=:Dkx^{\beta^*},  \quad D>0,
	\end{equation}
	\begin{align}\label{eq:Uh-example}
	\lambda_iU^{(r+\lambda_i)}h(x)=\begin{cases}
	\lambda_i\left(\frac{x}{\lambda_i+\delta_i}-\frac{K}{\lambda_i+r}\right), & K<x\le L,\\
	\frac{\lambda_i}{\lambda_i+r}(L-K), & x> L,
	\end{cases}
	\end{align}
	\begin{equation}\label{eq:Uu1-example}
	\lambda_2U^{(r+\lambda_2)}u_1(x)=\frac{\lambda_2 A\psi_{r+\lambda_1}(x)}{j_2(\gamma_{1,2})}+\frac{\lambda_2 B\varphi_{r+\lambda_1}(x)}{j_2(\gamma_{1,1})},
	\end{equation}
	and
	\begin{align}\label{eq:double-example}
	\lambda_2U^{(r+\lambda_2)}\lambda_1U^{(r+\lambda_1)}h(x)=\begin{cases}
	\lambda_2\lambda_1\left(\dfrac{x}{(\lambda_2+\delta_2)(\lambda_1+\delta_1)}-\dfrac{K}{(\lambda_2+r)(\lambda_1+r)}\right), & K<x\le L, \\
	\dfrac{\lambda_2\lambda_1}{(\lambda_2+r)(\lambda_1+r)}(L-K), & x> L.
	\end{cases}
	\end{align}
\end{subequations}
Note that by a direct calculation, we have $j_2(\gamma_{1,2})>0$ and $j_2(\gamma_{1,1})>0$ which confirms \eqref{eq:Uu1-example} is finite. \\
Figure \ref{fig:H2-entire-space} displays, in the case of $L=15$, the reward functions in Table \ref{tbl:example-rewards2} in the transformed space.  While we should wait until the next step completes to draw the functions in the right panel (b), we nonetheless show them here together with the left panel for convenience. The left panel (a) focuses on the region $y\in (0, F_2(c))$ and shows the reward function $h(x)-\lambda_2U^{(r+\lambda_2)}v^p_1(x)$ in the transformed space (which is $H_{23}(y)$) and the smallest nonnegative concave majorant $W_2(y)\equiv 0$. Hence, $\hat{u}_2(x)=0$ and
\[
v^*(x, 2)=\lambda_2U^{(r+\lambda_2)}v_1^p(x)=\lambda_2U^{(r+\lambda_2)}kx^{\beta^*}=Dkx^{\beta^*}, \quad x\in (0, c),
\] where $k$ and $c$ are given in \eqref{eq:k-and-c}.  It is confirmed that the region $x\in (0, c)$ is a continuation region for diffusion $X^{(2)}$.
\begin{center}
\begin{figure}[h]
	\begin{subfigure}{0.35\textwidth}
		\includegraphics[scale=0.4]{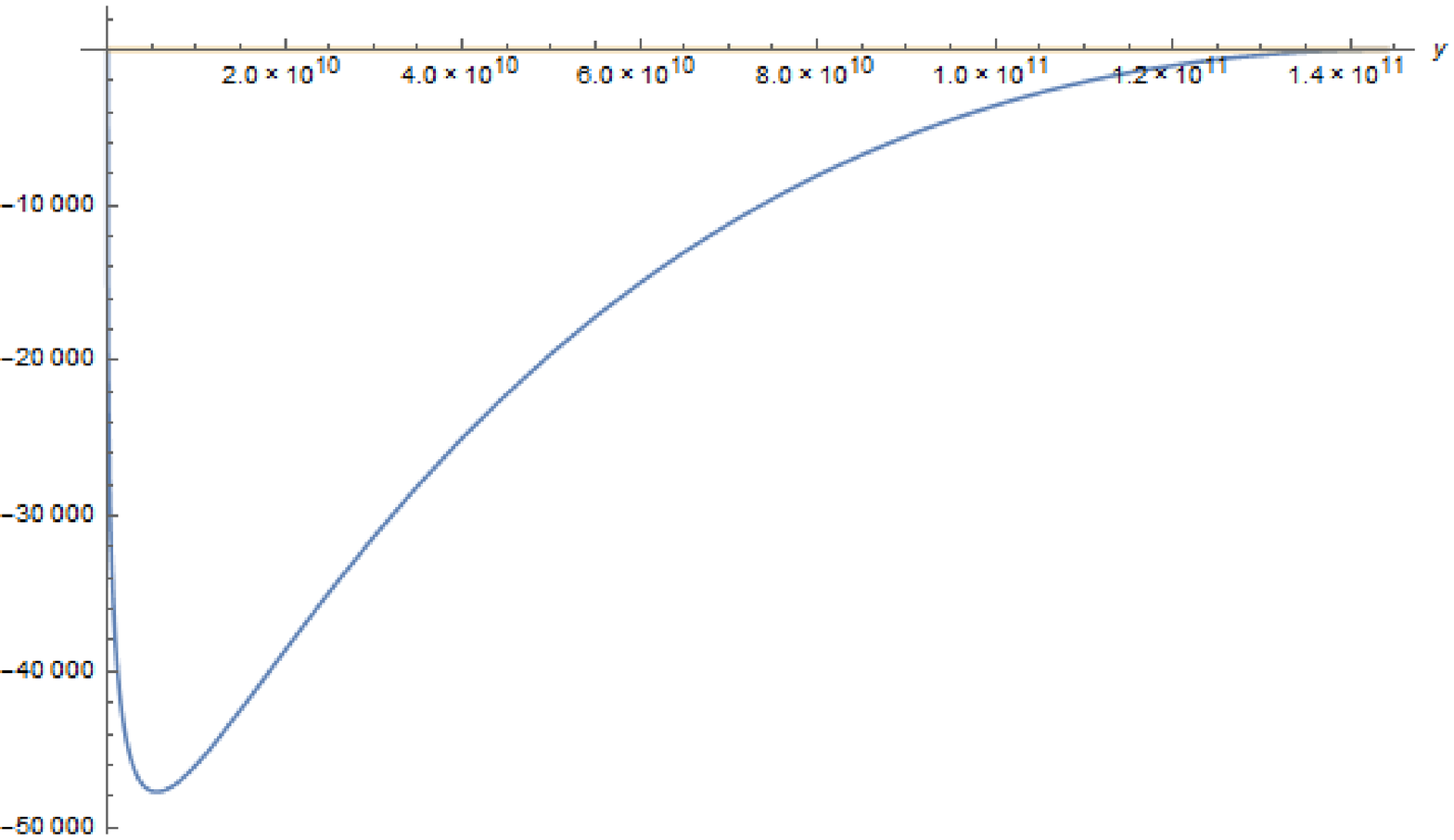}
		\caption{}
	\end{subfigure}\hspace{3cm}
	\begin{subfigure}{0.35\textwidth}
		\includegraphics[scale=0.55]{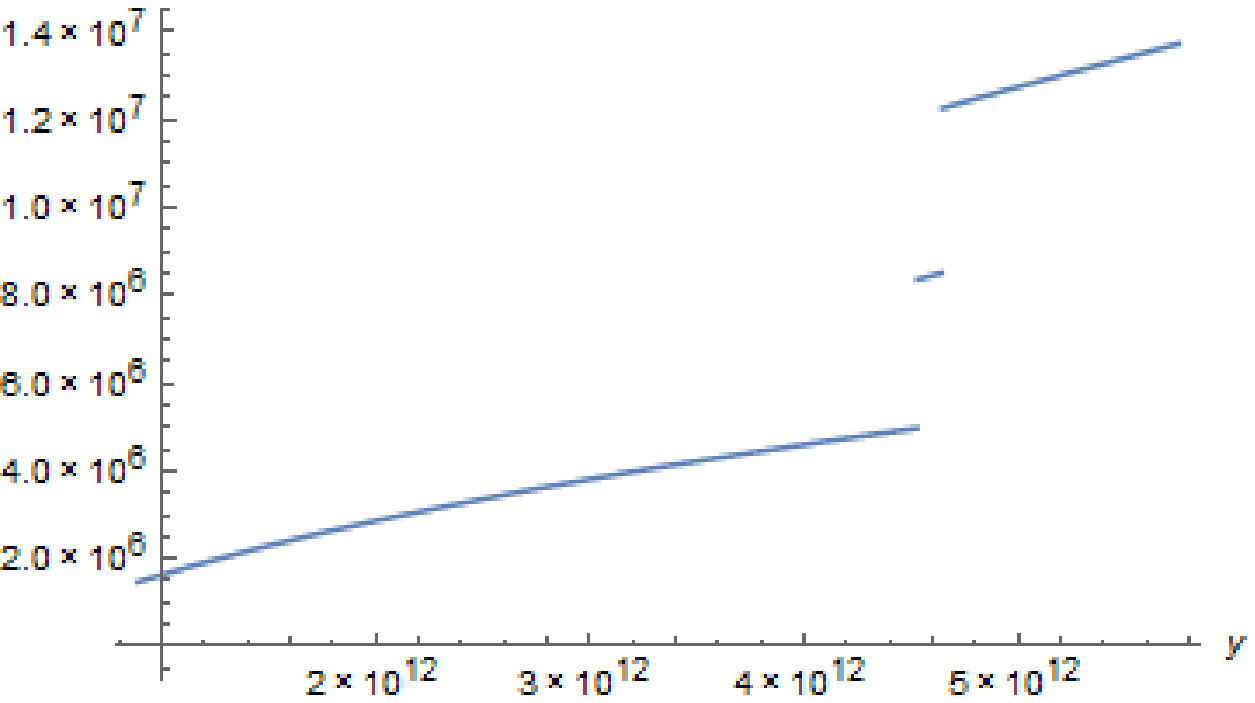}
		\caption{}
	\end{subfigure}
	\caption{The reward functions (Table \ref{tbl:example-rewards2}) in the transformed space by $F_2(x)$ when $L=15$. (a) The smallest nonnegative concave majorant $W_2(y)=0$ and the reward function $H_{23}(y)$ on $(0, F_2(c))$ and (b) the reward functions $H_{22}(y)$ and $H_{21}(y)$ on $[F_2(c), \infty)$, which are both concave; therefore, the smallest nonnegative concave majorants are these functions themselves.}
	\label{fig:H2-entire-space}
\end{figure}
\end{center}
\noindent\emph{Step (4)}: This step is concerned with identifying the threshold $a$. Due to the continuity of $v^*(x, 1)$ at $x=c$, we have $u_1(c)=kc^{\beta^*}-\lambda_1U^{(r+\lambda_1)}h(c)$ from \eqref{eq:condition1}. Therefore, we need to find the smallest nonnegative majorant  line $W_1(y)$ which passes $\left(F_1(c), \frac{u_1(c)}{\varphi_{r+\lambda_1}(c)}\right)$ and majorizes the first part of $H_{11}(y)$:
\begin{align*}
H_{11}(y)=\frac{h-\lambda_1 U^{(r+\lambda_1)}h}{\varphi_{r+\lambda_1}} \circ F_1^{-1}(y)
=\begin{cases}
\left(\frac{\delta_1}{\lambda_1+\delta_1}y^{\frac{1}{\theta_1}}-\frac{r}{\lambda_1+r}K\right)y^{-\frac{\gamma_{1,1}}{\theta_1}}, & y\in (F_1(c), F_1(L)],\\
\frac{r}{\lambda_1+r}(L-K)y^{-\frac{\gamma_{1,1}}{\theta_1}}, &y \in (F_1(L), \infty)
\end{cases}
\end{align*}
because the value function in the transformed space is the smallest nonnegative concave majorant of $H_{11}(y)$ on $[F_1(c), \infty)$ and $a\le L$ must hold. This function $W_1(F_1(x))$ fulfills the requirement (S-3) in Proposition \ref{prop:sufficient-cond}. It should be stressed that we find the  smallest nonnegative concave majorant of $H_{11}(y)$ on $[F_1(c), F_1(L))$. 

Refer to Table \ref{tbl:example-rewards1}. The reward at $x=c$ satisfies $h(c)-\lambda_1U^{(r+\lambda_1)}v^*(c, 2)=h(c)-\lambda_1U^{(r+\lambda_1)}h(c)$
since $c$ is in $\Gamma_2$.  Computing the difference of the preceding two equations at point $c$,
\[
kc^{\beta^*}-\lambda_1U^{(r+\lambda_1)}h(c)-(h(c)-\lambda_1U^{(r+\lambda_1)}h(c))=v^p_1(c)-h(c)>0
\] where the  positiveness is due to $c \in \mathrm{C}_1$ and the variational inequality. While this is true in all cases, we state it for this example as well. This means that when we transform the space by using $F_1(\cdot)$, the point $\frac{u_1(c)}{\varphi_{r+\lambda_1}(c)}$ is located above $H_{11}(F_1(c))$, enabling us to draw the smallest nonnegative concave majorant $W_1(y)$ from point $F_1(c)$. 

Note that the inflection point of the first part of $H_{11}(y)$ is $\hat{y}:=\left(\frac{r}{\delta_1}K\right)^{\theta_1}$ where we use the relationship $\frac{\gamma_{1, 1}\gamma_{1, 2}}{(1-\gamma_{1, 1})(1-\gamma_{1, 2})}=\frac{\lambda_1+r}{\lambda_1+\delta_1}$ and that
the second part of $H_{11}(y)$ is, by  direct differentiation,  a concave function.
When $L=15$, we have $c>\frac{r}{\delta_1}K=7.50$ and therefore, $H_{11}(y)$ is concave on $[F_2(c), \infty)$. Then the smallest nonnegative concave majorant of $H_{11}(y)$  is the line $W(y)=Ay+B$ with $A=0.0001278$ and $B=1436.5$, which is tangent to $H_{11}(y)$ at $F_1(a): a=14.9651$. It is depicted in Figure \ref{fig:H11}. 
\begin{figure}[h]
	\begin{center}
		\begin{minipage}{0.9\textwidth}
			\centering{\includegraphics[scale=0.8]{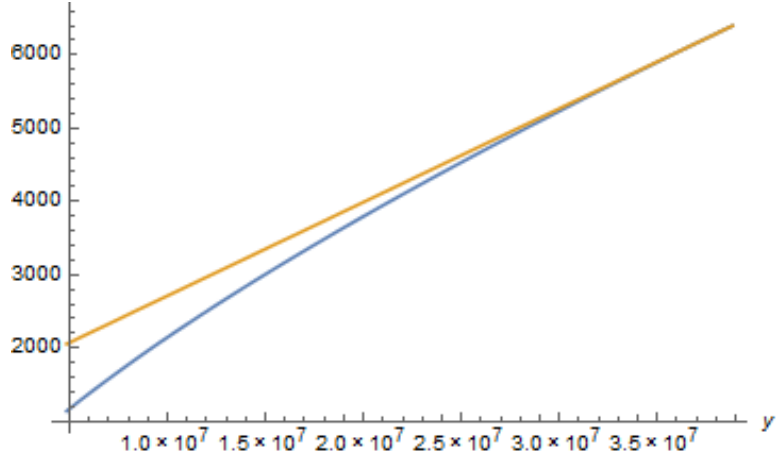}}\\
		\end{minipage}
		\caption{\small The graph of $H_{11}(y)$ on $[F_1(c), F_1(a))$ when $L=15$ and its smallest concave majorant $W_1(y)$.  } \label{fig:H11}
	\end{center}
\end{figure}
Once we have identified the point $a$, we go back to Panel (b) of Figure \ref{fig:H2-entire-space} which displays the reward function in Table \ref{tbl:example-rewards2} in the region $y\in [F_2(c), \infty)$ for $X^{(2)}$.  The discontinuities occur at $y=F_2(c)$, $y=F_2(a)$, and $y=F_2(L)$.  In $y\in [F_2(c), F_2(a))$, the function is
\[H_{22}(y)=\frac{h-\lambda_2U^{(r+\lambda_2)}[u_1+\lambda_1U^{(r+\lambda_1)}h] }{\varphi_{r+\lambda_2}} \circ F_2^{-1}(y),
\] where the specific forms of the functions in the numerator are shown in \eqref{eq:Uu1-example} and \eqref{eq:double-example}. In $y\in [F_2(a), \infty)$, the function is
\[
H_{21}(y)=\frac{h-\lambda_2 U^{(r+\lambda_2)}h}{\varphi_{r+\lambda_2}} \circ F_2^{-1}(y),
\] where the specific form of the function in the numerator is given in \eqref{eq:Uh-example}. Since both $H_{21}(y)$ and $H_{22}(y)$ are concave in their respective domains, the smallest nonnegative concave majorants are the functions themselves as desired (see S-4 in Proposition \ref{prop:sufficient-cond}).  That is, $[c, \infty)$ is the stopping region for diffusion $X^{(2)}$.

To confirm that condition S-1 (Proposition \ref{prop:sufficient-cond}) is satisfied, we need to draw
\[
H_{12}(y):=\frac{h-\lambda_1\lambda_2U^{(r+\lambda_1)}U^{(r+\lambda_2)}v^p_1}{\varphi_{r+\lambda_1}}\circ F_{1}^{-1}(y)=\frac{h-v^p_1}{\varphi_{r+\lambda_1}}\circ F_{1}^{-1}(y), \quad y\in (0, F_1(c))
\]
and see if it is negative on $(0, F_1(c))$ so that $H_{12}(y)$ is majorized by $W_1(y)\equiv 0$ in $A_4$. Figure \ref{fig:H1-entire} illustrates, in the case of $L=15$, the reward functions in Table \ref{tbl:example-rewards1} in the transformed space. The first panel (a) shows $H_{12}(y)$ on $(0, F_1(K))$ and the second panel does the same on $[F_1(K), F_1(c))$.  That is, the two panels show the function  $H_{12}(y)$ on $(0, F_1(c))$ and its smallest nonnegative concave majorant $W_1(y)\equiv 0$ in the region. The third panel (c) focuses on the region $[F_1(c), \infty)$. The reward function is $H_{11}(y)$ in the transformed space. 
\begin{center}
	\begin{figure}[h]
		\begin{subfigure}{0.45\textwidth}
			\includegraphics[scale=0.6]{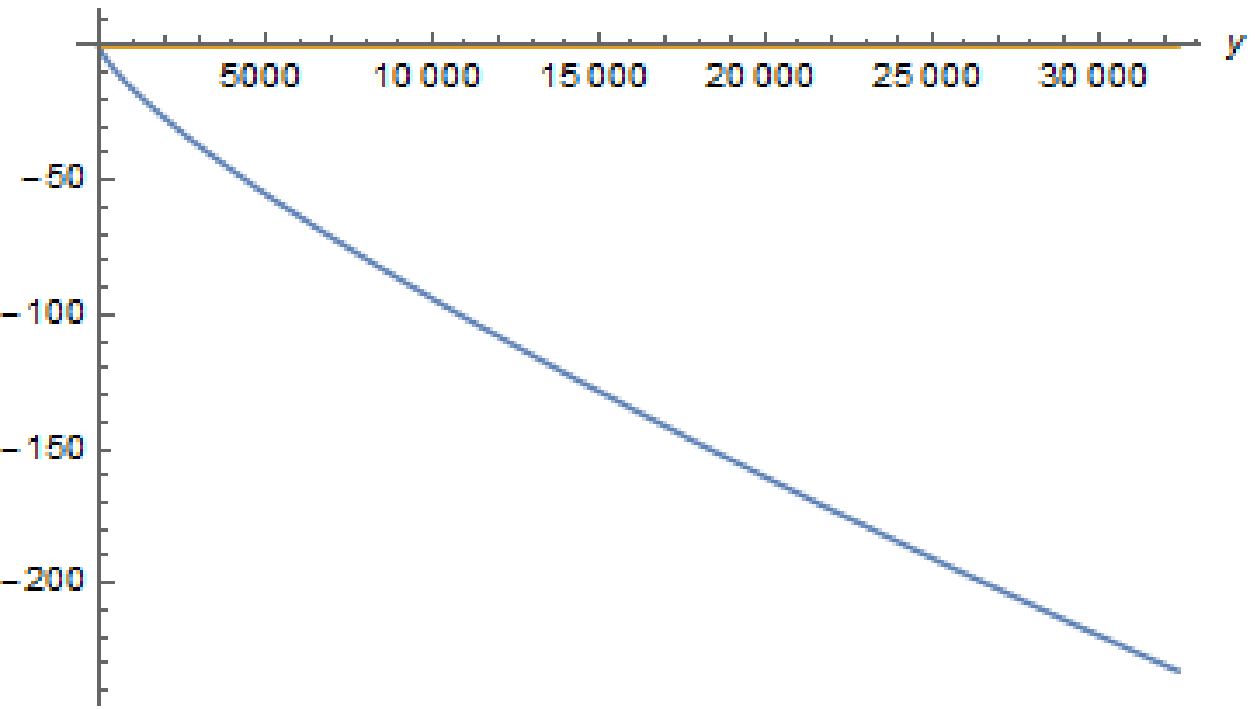}
			\caption{}
		\end{subfigure}
		\begin{subfigure}{0.45\textwidth}
			\includegraphics[scale=0.6]{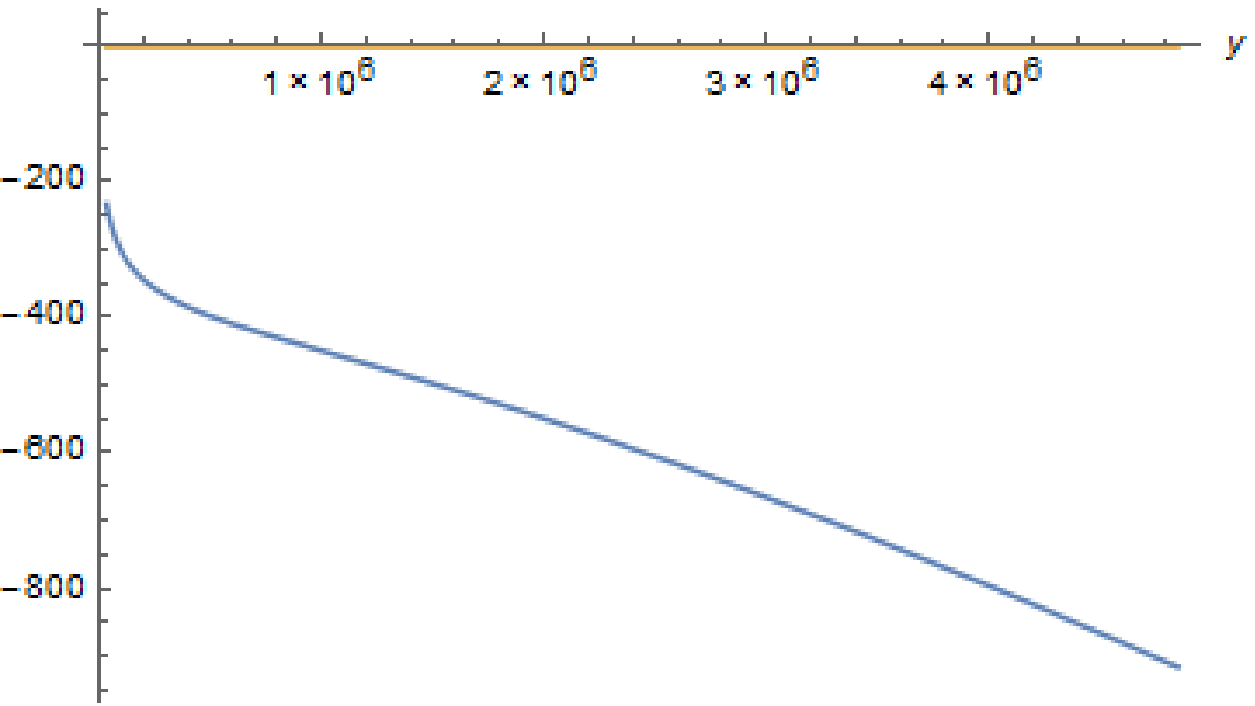}
			\caption{}
		\end{subfigure}
		\begin{subfigure}{0.45\textwidth}
			\includegraphics[scale=0.6]{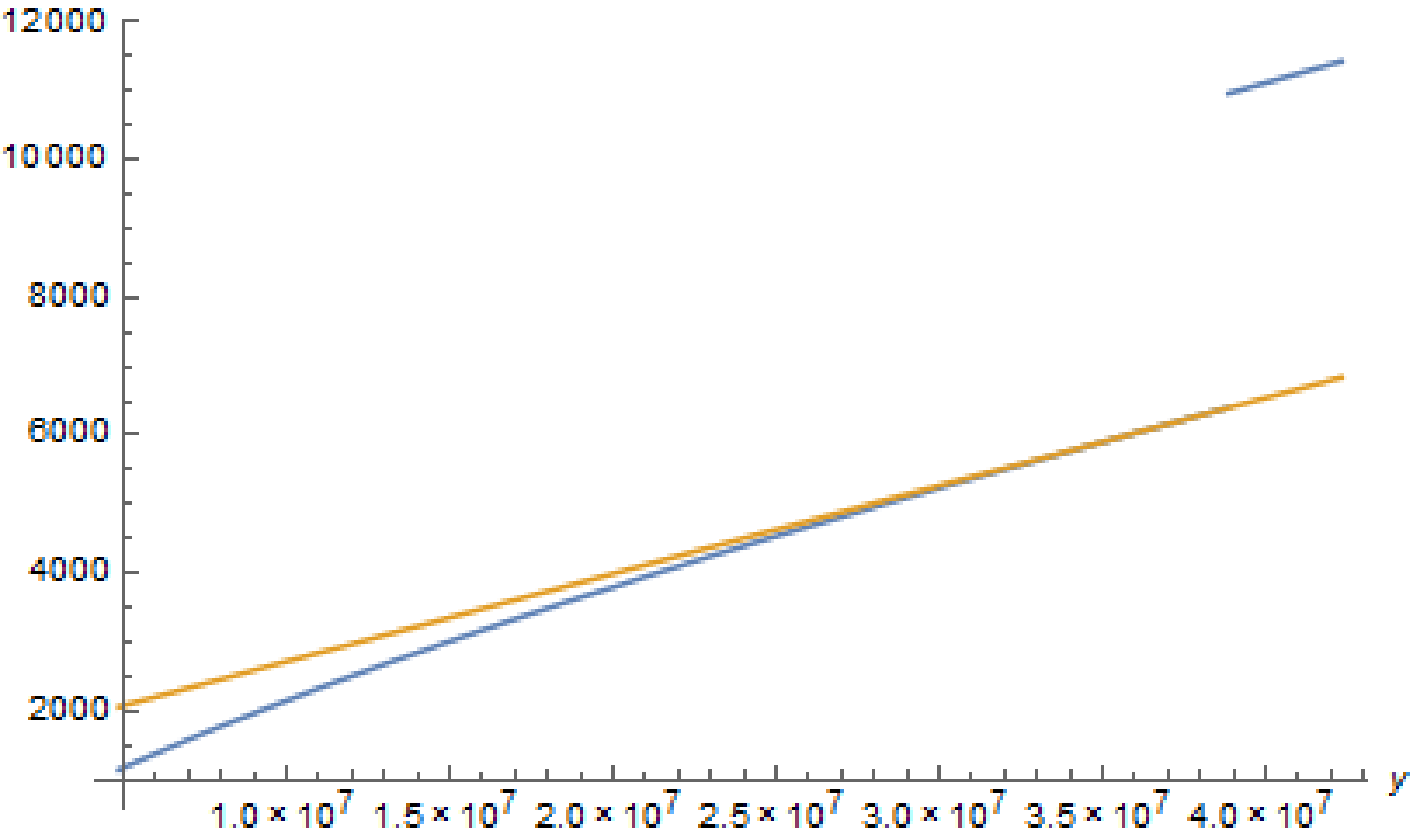}
			\caption{}
		\end{subfigure}
		\caption{The reward functions in Table \ref{tbl:example-rewards1} in the transformed space by $F_1(x)$ when $L=15$. (a) The smallest nonnegative concave majorant $W_1(y)=0$ and the reward function $H_{12}(y)$ on $(0,  F_1(K))$ and (b) the same on $[F_1(K), F_1(c))$.  The third panel (c) depicts the line tangent to $H_{11}(y)$ on $[F_1(c), \infty)$. }
		\label{fig:H1-entire}
	\end{figure}
\end{center}
Note that the line tangent to $H_{11}(y)$ is the same  as the one shown in Figure \ref{fig:H11}. Hence, the smallest concave majorant $W_1(y)$ is zero  on $(0, F_1(c))$, the line $Ay+B$ on $[F_1(c), F_1(a))$ and $H_{11}(y)$ itself on $[F_1(a), \infty)$.  The discontinuity of $H_{11}(y)$ corresponds to point $F_1(L)$.
The continuation region $\mathrm{C}_1$ is $(0, c)\cup [c, a)$ and the stopping region is $[a, \infty)$.  Note further that if there is no tangency point on $(F_1(c), F_1(L))$, then the smallest nonnegative concave majorant $W_1(y)$ is the line connecting $\left(F_1(c), \frac{u_1(c)}{\varphi_{r+\lambda_1}(c)}\right)$ and $(F_1(L), H_{11}(F_1(L)))$. In this case, as well, we can find the slope $A$ and the intercept $B$ with $a=L$ (see Section \ref{subsec:binding} where $L=11.3$).

Now from Tables \ref{tbl:example-rewards1} and \ref{tbl:example-rewards2}, we obtain 
\begin{equation}\label{v2-example}
v^*(x, 2)=\begin{cases}
Dkx^{\beta^*}, & x \in (0, c),\\
x\wedge L-K, & x\in [c, \infty).
\end{cases}
\end{equation}
\begin{equation}\label{v1-example}
v^*(x, 1)=\begin{cases}
kx^{\beta^*}, & x \in (0, c),\\
A\psi_{r+\lambda_1}(x)+B\varphi_{r+\lambda_1}(x)+\lambda_1U^{(r+\lambda_1)}(x-K), & x\in [c, a),\\
x\wedge L-K, & x\in [a, \infty).
\end{cases}
\end{equation}
for the case when $L=15$. Figure \ref{fig:v1-v2} summarizes the functions $v^*(x, 1), v^*(x, 2)$, and $h(x)$ from \eqref{v2-example} and \eqref{v1-example} on $(0, \infty)$ when $L=15$. We have $v^*(x, 1)\ge v^*(x, 2)\ge h(x)$ on $x\in (0, \infty)$. We see that $v^*(x, 1)$ is smooth at $a=14.9651$ and $v^*(x, 2)$ is smooth at $c=10.8661$. Note  that we have checked all the sufficient conditions S-1, S-2, S-3, S-4 in Proposition \ref{prop:sufficient-cond}.
\begin{figure}[]
	\begin{center}
		\begin{minipage}{0.9\textwidth}
			\centering{\includegraphics[scale=0.8]{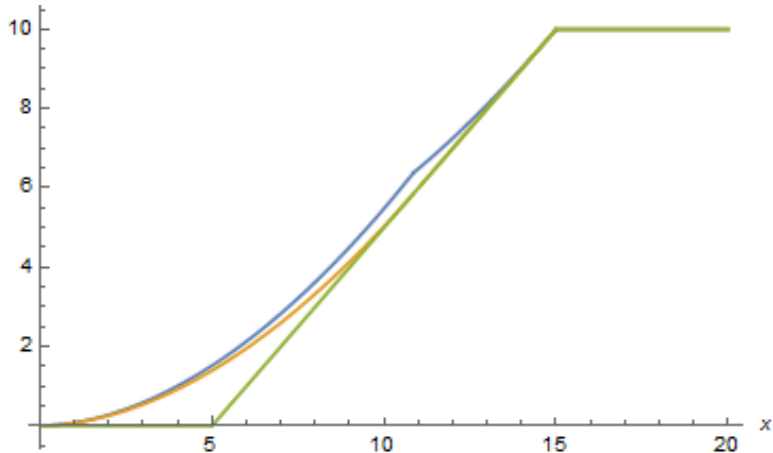}}\\
		\end{minipage}
		\caption{\small The graphs of $v^*(x, 1)$ (topmost, blue), $v^*(x, 2)$ (middle, orange),  and $h(x)$ (bottom, green) when $L=15$. } \label{fig:v1-v2}
	\end{center}
\end{figure}
\subsection{Binding Constraint}\label{subsec:binding}
As $L$ becomes small, we need to make sure that the resulting pair $(v^*(x,1),v^*(x,2))$ in $A_4$ region gives appropriate values. Note that when $X$ starts at a point in $(0, L)$, the point $L$ serves as the absorbing boundary. The solution to \eqref{eq:k-and-c} may yield $c\ge L$. Even if we can find $K<c\le L$ from \eqref{eq:k-and-c}, with small $L$ we may encounter a case when  $v^*(c,1)>L-K$. This is not surprising because \eqref{eq:k-and-c} does not take the constraint $L$ into consideration. The equations in \eqref{eq:k-and-c} can be seen as necessary conditions for unconstrained optimization. When the constraint $L$ becomes binding, we cannot ensure that $H'_{23}(F_2(c))=0$ (see \eqref{eq:1-2}) but $H_{23}(F_2(c))=0$ still holds since $v^*(x,2)$ is continuous (see Proposition \ref{prop:necessary_cond}). 

Let $p\le L$ be the point that $H_{23}(F_2(p))=0$. The weight $k$ satisfying this equation can be expressed as a function of $p$ using \eqref{eq:Uv1-example}: $k(p)=\left(\frac{p-K}{D}\right)p^{-\beta^*}$. Then we define
\begin{align*}
v_2(x, 2;p):=k(p)Dx^{\beta^*},
\end{align*}
a family of functions varying with the values of $p$.  Since $k(p)\ge 0$ for any $p\in [K, L]$, $v_2(x, 2; p)\ge v_2(x, 2; p')$ for any $x\in (0, L]$ if and only if $k(p)\ge k(p')$.  Let us take the derivative of $k(p)$ with respect to $p$
\[
\frac{\diff}{\diff p}k(p)=p^{-\beta^*-1}\frac{((1-\beta^*)p+\beta^*K)}{D}
\] which is nonnegative because we have $\beta^*\le \frac{p}{p-K}$: note that we have assumed that \eqref{eq:c} is violated.  It follows that $v_2(x,2;p)$ is increasing in $p$. Therefore, we shall choose the largest value of $p$ that assures a feasible solution. Furthermore, $p$ should be such that $v^*(x,1)\ge v^*(y,1)$ holds for $x\ge y$, $x,y\in A_2$. This means that $v^*(x,1)$ is nondecreasing on $A_2$ since the stopping region of $X^{(1)}$ is on the right of $A_2$ and $h$ is nondecreasing. Once we set $c$ to a certain value, it will give us $k(c)$. Using $(c,k(c))$, we can employ the same procedure as in \emph{Step (4)} above to find $u_1$. We will demonstrate this by an example.
\newline\indent Let us set $L=11.3$. If we use the conditions in \eqref{eq:1-2}, we compute $c=10.8661$ and $v^*(c,1)=6.3943>L-K$ violates the constraint. The largest value of $c$ satisfying the conditions in the preceding paragraph is $c=10.7600$. Then, $k=0.0770$ and $\frac{u_1(c)}{\phi_{r+\lambda_1}(c)}=1948.21$. Note that the conditions (S-1) and (S-2) in Proposition \ref{prop:sufficient-cond} are satisfied. The smallest nonnegative concave majorant of $H_{11}$ will be the line passing $\left(F_1(c), \frac{u_1(c)}{\varphi_{r+\lambda_1}(c)}\right)$ and $(F_1(L), H_{11}(F_1(L)))$ (see Figure \ref{fig:u1constrained}). The slope of this line is $-0.0003062$. Thus, $a=L$ and
\begin{equation}\label{v1-example-singular}
v^*(x, 1)=\begin{cases}
kx^{\beta^*}, & x \in (0, c),\\
A'\psi_{r+\lambda_1}(x)+B'\varphi_{r+\lambda_1}(x)+\lambda_1U^{(r+\lambda_1)}(x-K), & x\in [c, L),\\
L-K, & x\in [L, \infty),
\end{cases}
\end{equation}
where $A'$ and $B'$ are the slope and the intercept of the line connecting $\left(F_1(c), \frac{u_1(c)}{\varphi_{r+\lambda_1}(c)}\right)$ and $(F_1(L), H_{11}(F_1(L)))$. Note that the condition (S-3) in Proposition \ref{prop:sufficient-cond} is satisfied. The condition (S-4) in Proposition \ref{prop:sufficient-cond} is also satisfied and $[c,\infty)$ is the stopping region for $X^{(2)}$. The resulting value functions are displayed in Figure \ref{fig:Vsconstrained}.
\begin{figure}[h]
	\includegraphics[scale=0.6]{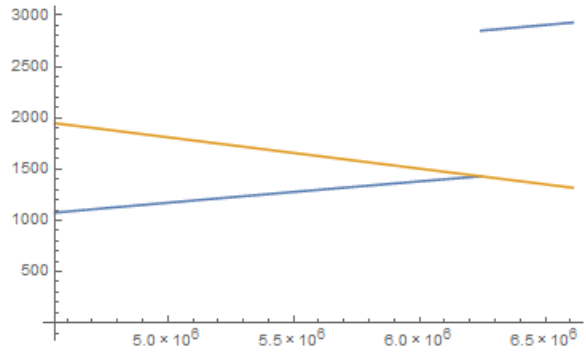}
	\caption{\small The graph of $H_{11}$ and the line  passing $\left(F_1(c), \frac{u_1(c)}{\varphi_{r+\lambda_1}(c)}\right)$ and $(F_1(L), H_{11}(F_1(L)))$ for $L=11.3$. Since the point $L$ is the absorbing boundary, we only need to consider the region on $(0,F_1(L))$.}\label{fig:u1constrained}
\end{figure}
\begin{figure}[h]
	\includegraphics[scale=0.65]{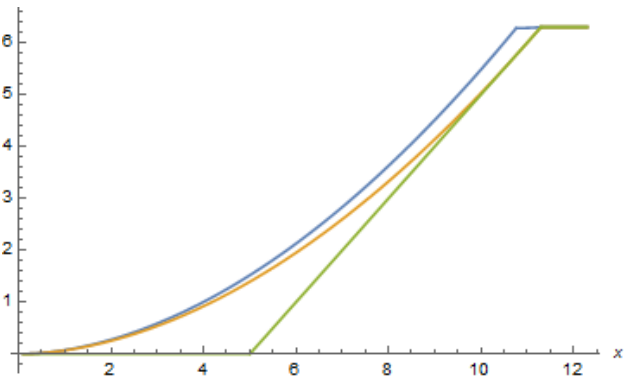}
	\caption{\small The graphs of $v^*(x, 1)$ (topmost, blue), $v^*(x, 2)$ (middle, orange)  and $h(x)$ (bottom, green) when $L=11.3$. Here $v^*(c,1)=6.2786<L-K$. From $c$ to $a=L$, $v^*(x,1)$ is increasing.}\label{fig:Vsconstrained}
\end{figure}

\remark{\normalfont In the end, we comment that it is possible to extend the idea of the paper to the model with more than two (any finite number) regimes given the matrix of transition probabilities. When starting in regime $i$, the second term in the DPP equation (Lemma \ref{lem:hitting}) will involve a weighted sum of $\bar{v}(\cdot,j)$ ($j\neq i$) weighted by the indicator that the regime was switched to regime $j$. While the increased number of regimes will lead to greater number of regions $A_j$, the idea of finding the value functions in each $A_j$ is the same as described in Sections \ref{subsec:form} and \ref{sec:method}.}


\begin{thebibliography}{22}
	\providecommand{\natexlab}[1]{#1}
	\providecommand{\url}[1]{\texttt{#1}}
	\expandafter\ifx\csname urlstyle\endcsname\relax
	\providecommand{\doi}[1]{doi: #1}\else
	\providecommand{\doi}{doi: \begingroup \urlstyle{rm}\Url}\fi
	
	\bibitem[Babbin et~al.(2014)Babbin, Forsyth, and Labahn]{babbin2014}
	J.~Babbin, P.~A. Forsyth, and G.~Labahn.
	\newblock A comparison of iterated optimal stopping and local policy iteration
	for {A}merican options under regime switching.
	\newblock \emph{Journal of Scientific Computing}, 58\penalty0 (2):\penalty0
	409--430, 2014.
	
	\bibitem[Boyarchenko and Levendorski{\v i}(2009)]{boyarchenko2009}
	S.~Boyarchenko and S.~Levendorski{\v i}.
	\newblock American options in regime-switching models.
	\newblock \emph{SIAM J. Control Optim.}, 48\penalty0 (3):\penalty0
	1353–--1376, 2009.
	
	\bibitem[Broadie and Detemple(1995)]{broadie-detemple1995}
	M.~Broadie and J.~Detemple.
	\newblock American capped call options on dividend-paying assets.
	\newblock \emph{Review of Financial Studies}, 8\penalty0 (1):\penalty0
	161--191, 1995.
	
	\bibitem[\c{C}inlar(2011)]{cinlar}
	E.~\c{C}inlar.
	\newblock \emph{Probability and Stochastics}, volume 261 of \emph{Graduate
		Texts in Mathematics}.
	\newblock Springer Science+Business Media, LLC, 2011.
	
	\bibitem[Dayanik and Karatzas(2003)]{DK2003}
	S.~Dayanik and I.~Karatzas.
	\newblock On the optimal stopping problem for one-dimensional diffusions.
	\newblock \emph{Stochastic Process. Appl.}, 107 (2):\penalty0 173--212, 2003.
	
	\bibitem[Dynkin(1965)]{dynkin}
	E.~B. Dynkin.
	\newblock \emph{Markov Processes II}.
	\newblock Springer, Berlin Heidelberg, 1965.
	
	\bibitem[Guo(2001)]{guo2001}
	X.~Guo.
	\newblock An explicit solution to an optimal stopping problem with regime
	switching.
	\newblock \emph{J. Appl. Probab.}, 38\penalty0 (2):\penalty0 464--481, 2001.
	
	\bibitem[Guo and Zhang(2004)]{guo2004}
	X.~Guo and Q.~Zhang.
	\newblock Closed-form solutions for perpetual {A}merican put options with
	regime switching.
	\newblock \emph{SIAM J. Appl. Math.}, 64\penalty0 (6):\penalty0 2034--2049,
	2004.
	
	\bibitem[Hamilton(2008)]{palgrave_dictionary}
	J.~D. Hamilton.
	\newblock Regime switching models.
	\newblock In S.~N. Durlauf and L.~E. Blume, editors, \emph{The New Palgrave
		Dictionary of Economics}. Palgrave Macmillan, 2nd Edition, 2008.
	
	\bibitem[Huang et~al.(2011)Huang, Forsyth, and Labahn]{labahn2011}
	Y.~Huang, P.~A. Forsyth, and G.~Labahn.
	\newblock Methods for pricing {A}merican options under regime switching.
	\newblock \emph{SIAM J. Sci. Comput.}, 33\penalty0 (5):\penalty0 2144–--2168,
	2011.
	
	\bibitem[It\^{o} and McKean(1974)]{IM1974}
	K.~It\^{o} and H.~P. McKean, Jr.
	\newblock \emph{Diffusion Processes and their Sample Paths}.
	\newblock Springer, Berlin Heidelberg, 1974.
	
	\bibitem[Karatzas and Shreve(1991)]{karatzas-shreve-book1}
	I.~Karatzas and S.~E. Shreve.
	\newblock \emph{Browninan Motion and Stochastic Calculus, 2nd Edition}.
	\newblock Springer-Verlag, New York, 1991.
	
	\bibitem[Le and Wang(2010)]{le-wang}
	H.~Le and C.~Wang.
	\newblock A finite time horizon optimal stopping problem with regime switching.
	\newblock \emph{SIAM J. Control Optim.}, 48\penalty0 (8):\penalty0
	5193--–5213, 2010.
	
	\bibitem[Lebedev(1972)]{lebedev}
	N.~N. Lebedev.
	\newblock \emph{Special Functions and their Applications, revised edn}.
	\newblock Dover Publications, New York, 1972.
	
	\bibitem[Liu(2016)]{liu2016}
	R.H. Liu.
	\newblock Optimal stopping of switching diffusions with state dependent
	switching rates.
	\newblock \emph{Stochastics}, 88\penalty0 (4):\penalty0 586--605, 2016.
	
	\bibitem[Mao and Yuan(2006)]{markovian_sde}
	X.~Mao and C.~Yuan.
	\newblock \emph{Stochastic Differential Equations with Markovian Switching}.
	\newblock Imperial College Press, 2006.
	
	\bibitem[{\O}ksendal(2003)]{oksendal-book}
	B.~{\O}ksendal.
	\newblock \emph{Stochastic Differential Equations, An Introduction with
		Applications, 6th edition}.
	\newblock Springer, 2003.
	
	\bibitem[Pemy(2014)]{pemy2014}
	M.~Pemy.
	\newblock Optimal stopping of {M}arkov switching {L}{\'e}vy processes.
	\newblock \emph{Stochastics}, 86\penalty0 (2):\penalty0 341--369, 2014.
	
	\bibitem[Pham(2007)]{pham2007}
	H.~Pham.
	\newblock On the smooth-fit property for one-dimensional optimal switching
	problem.
	\newblock \emph{S\'{e}minaire Probabilit\'{e}s}, XL:\penalty0 187--199, 2007.
	
	\bibitem[Pham(2009)]{Pham-book}
	H.~Pham.
	\newblock \emph{Continuous-time Stochastic Control and Optimization with
		Financial Applications}, volume~61 of \emph{Stochastic Modelling and Applied
		Probability}.
	\newblock Springer, Berlin Heidelberg, 2009.
	
	\bibitem[Royden and Fitzpatrick(2010)]{royden_real_analysis}
	H.~L. Royden and P.~M. Fitzpatrick.
	\newblock \emph{Real Analysis, 4th edition}.
	\newblock Prentice Hall, 2010.
	
	\bibitem[Rudin(1976)]{rudin_math_analysis}
	W.~Rudin.
	\newblock \emph{Principles of Mathematical Analysis, 3rd edition}.
	\newblock McGraw-Hill, Inc., 1976.
	
\end{thebibliography}
\def\cprime{$'$}

\end{document}